%% file: main.tex
\documentclass{article}

\usepackage{amsmath,amssymb,amsthm}
\usepackage{xparse}
\usepackage[dvipsnames]{xcolor}
\usepackage{tikz}
\usepackage{bbold}
\usepackage{caption}
\usepackage{enumitem}
\usepackage{multirow}
\usepackage{fullpage}

\usepackage{fourier-orns}
\usepackage{fourier}

\DeclareMathOperator{\e}{e}
\DeclareMathOperator*{\argmin}{arg\,min}
\DeclareMathOperator*{\degr}{deg}
\DeclareMathOperator*{\height}{height}
\DeclareMathOperator*{\DAG}{\mathcal{R}}
\DeclareMathOperator*{\roottree}{root}
\DeclareMathOperator*{\child}{children}

\DeclareMathOperator*{\tree}{\mathcal{R}^{-1}}

\begin{document}

\input{tikz_costflow}
\input{tikz_thetaT}
\input{tikz_minidags}
\input{tikz_algo-ed}
\input{tikz_exdag}
\input{tikz_ex-worstcase}

\newtheorem{proposition}{Proposition}
\newtheorem{lemma}{Lemma}

\title{Approximation of trees by self-nested trees}

\author{Romain Aza\"{\i}s\thanks{Laboratoire Reproduction et D\'eveloppement des Plantes, Univ Lyon, ENS de Lyon, UCB Lyon 1, CNRS, INRA, Inria, F-69342, Lyon, France.} \\
\and
Jean-Baptiste Durand\thanks{Laboratoire Jean Kuntzmann, MISTIS, INRIA Grenoble -- Rh\^one-Alpes, Saint Ismier, France.}
\and
Christophe Godin\footnotemark[1]}

\date{}

\maketitle

\begin{abstract} \small\baselineskip=9pt The class of self-nested trees presents remarkable compression properties because of the systematic repetition of subtrees in their structure. In this paper, we provide a better combinatorial characterization of this specific family of trees. In particular, we show from both theoretical and practical viewpoints that complex queries can be quickly answered in self-nested trees compared to general trees.
We also present an approximation algorithm of a tree by a self-nested one that can be used in fast prediction of edit distance between two trees.\end{abstract}

\section{Introduction}
\label{s:intro}

Trees form an expanded family of combinatorial objects that offers a wide range of application fields, from plant modeling to XML files analysis through study of RNA secondary structure. Complex queries on tree structures (e.g., computation of edit distance, finding common substructures, compression) are required to handle these models. A critical question is to control the complexity of the algorithms implemented to solve these queries. One way to address this issue is to approximate the original trees by simplified structures that achieve good algorithmic properties. One can expect good algorithmic properties from structures that present a high level of redundancy in their substructures. Indeed, one can take account these repetitions to avoid redundant computations on the whole structure.

Searching redundancies in the tree often allows to design efficient compression methods. As it is explained in \cite{Bille2015166}, one often considers the following two types of repeated substructures: subtree repeat (used in DAG compression \cite{BM,Buneman:2003:PQC:1315451.1315465,frick2003query,GF2010}) and tree pattern repeat (exploited in tree grammars \cite{Busatto:2008:EMR:1370308.1370503,Lohrey:2006:CTA:1217607.1217615,5749493} and top tree compression \cite{Bille2015166}). A survey on this topic may be found in \cite{Sakr2009303} in the context of XML files. In this paper, we restrict ourselves to DAG compression, which consists in building a Directed Acyclic Graph (DAG) that represents a tree without displaying the redundancy of its identical subtrees. Previous algorithms have been proposed to allow the computation of the DAG of an ordered tree with complexities ranging in $O(n^2)$ to $O(n)$ \cite{Downey:1980:VCS:322217.322228}, where $n$ is the number of vertices of the tree. In the case of unordered trees, two different algorithms exist \cite[2.2 Computing Tree Reduction]{GF2010}, that share the same time-complexity in $O(n^2\times d\times\log(d))$, where $n$ is the number of vertices of the tree and $d$ denotes its outdegree. From now on, we limit ourselves to unordered trees.

Trees that are the most compressed by DAG compression scheme present the highest level of redundancy in their subtrees: all the subtrees of a given height must be isomorphic. In this case, regardless of the number of vertices, the DAG related to a tree $\tau$ has exactly $H+1$ vertices, where $H$ denotes the height of $\tau$, which is the minimal number of vertices that may be reached. This family of trees has been introduced in \cite{greenlaw96} under the name of nested trees as an interesting class of trees for which the subtree isomorphism problem is in $\text{NC}^2$. Later, they have been called self-nested trees \cite[Definition 7]{GF2010} to insist on their recursive structure and their proximity to the notion of self-similarity.

In this article, we prove the algorithmic efficiency of self-nested trees through different questions (compression, evaluation of recursive functions, evaluation of edit distance) and study their combinatorics. In particular, we establish that self-nested trees are roughly exponentially less frequent than general trees. This combinatorics can be an asset in exhaustive search pro\-blems. Nevertheless, this result also says that one can not always take advantage of the remarkable algorithmic properties of self-nested trees when working with ge\-ne\-ral trees. Consequently, our aim is to investigate how general unordered trees can be approximated by simplified trees in the class of self-nested trees from both theoretical and numerical perspectives. Our objective is to take advantage of the aforementioned qualities of self-nested trees even for a tree that is not self-nested. In particular, we show that our approximation algorithm can be used to very rapidly predict the edit distance between two trees, which is a usual but costly operation for comparing tree data in computational biology for instance \cite{doi:10.1093/bioinformatics/6.4.309}.

The paper is organized as follows. Section \ref{s:prelim} is devoted to the presentation of the concepts of interest in this paper, namely unordered trees, self-nested trees and tree reduction. Algorithmic efficiency of self-nested trees is investigated in Section \ref{s:snt}. Combinatorial properties of self-nested trees are presented in Section \ref{section:properties}. Our approximation algorithm is developed in Section \ref{s:approx}. Section \ref{s:fastpred} is dedicated to an application to fast prediction of edit distance. All the proofs have been deferred to the supplementary file.

\paragraph{Note on numerical results} We ran all the si\-mu\-lations of the paper in \verb+Python3+ on a Macbook Pro laptop running OSX High Sierra, with 2.9 GHz Intel Core i7 processors and 16 GB of RAM.


\section{Preliminaries}
\label{s:prelim}

This section is devoted to the precise formulation of the structures of interest in this paper, among which the class of unordered rooted trees $\mathbb{T}$, the set of self-nested trees $\mathbb{T}^{sn}$ and the concept of tree reduction.

\subsection{Unordered rooted trees}

A rooted tree $\tau=(V,E)$ is a connected digraph containing no cycle and such that there exists a unique vertex called $\roottree(\tau)$ which has no parent. Any vertex different from the root has exactly one parent. For any vertex $v$ of $\tau$, $\child(v)$ denotes the set of vertices that have $v$ as parent. The leaves of $\tau$ are all the vertices without children. The height of a vertex $v$ may be recursively defined as $\height(v)=0$ if $v$ is a leaf of $\tau$ and
$$\height(v)=1+\max_{w\in\child(v)}\height(w)$$
otherwise. The height of the tree $\tau$ is defined as the height of its root, $\height(\tau)=\height(\roottree(\tau))$. The outdegree $\deg(\tau)$ of $\tau$ is the maximal branching factor that can be found in $\tau$, that is
$$\deg(\tau)=\max_{v\in\tau}\#\child(v),$$
where $\#A$ denotes the cardinality of the set $A$. With a slight abuse of notation, $\#\tau$ denotes the number of vertices of $\tau$. A subtree $\tau[v]=(V[v],E[v])$ rooted in $v$ is the connected subgraph of $\tau$ such that $V[v]$ is the set of the descendants of $v$ in $\tau$ and $E[v]$ is defined as
$$E[v]=\left\{(\xi,\xi')\in E~:~\xi\in V[v],\,\xi'\in V[v]\right\}.$$

In the sequel, we consider unordered rooted trees for which the order among the sibling vertices of any vertex is not significant. A precise characterization is obtained from the additional definition of isomorphic trees. Let $\tau_1=(V_1,E_1)$ and $\tau_2=(V_2,E_2)$ be two rooted trees. A one-to-one correspondence $\varphi:V_1\to V_2$ is called a tree isomorphism if, for any edge $(v,w)\in E_1$, $(\varphi(v),\varphi(w))\in E_2$. Structures $\tau_1$ and $\tau_2$ are called isomorphic trees whenever there exists a tree isomorphism between them. One can determine if two $n$-vertex trees are isomorphic in $O(n)$ \cite[Example 3.2 and Theorem 3.3]{Aho:1974:DAC:578775}. The existence of a tree isomorphism defines an equivalence relation $\equiv$ on the set of rooted trees. The class of unordered rooted trees is the set of equivalence classes for this relation, i.e., the quotient set of rooted trees by the existence of a tree isomorphism. We refer the reader to \cite[I.5.2. Non-plane trees]{FS2009} for more details on this combinatorial class. From now on, all the trees are unordered rooted trees.

\paragraph{Simulation of random trees} We describe here the algorithm that we used in this paper to generate random trees of given size. From a tree with $n-1$ vertices, we construct a tree of size $n$ by adding a child to a randomly chosen vertex (uniform distribution). We point out that the position of the new child is not significant since trees are unordered. Starting from the tree composed of a unique vertex, we repeat the procedure $n-1$ times to obtain a random tree of size $n$.

\subsection{Tree reduction}
\label{ss:treereduction}

Let us now consider the equivalence relation $\equiv$ on the set of the subtrees of a tree $\tau=(V,E)$. We consider the quotient graph $Q=(V_\equiv,E_\equiv)$ obtained from $\tau$ using this equivalence relation. $V_\equiv$ is the set of equivalence classes on the subtrees of $\tau$, while $E_\equiv$ is a set of pairs of equivalence classes $(C_1,C_2)$ such that the root of $C_2$ is a child of the root of $C_1$ (modulo isomorphism). In light of \cite[Proposition 1]{GF2010}, the graph $Q$ is a Directed Acyclic Graph (DAG), that is a connected digraph without path from any vertex $x$ to itself. Let $(C_1,C_2)$ be an edge of the DAG $Q$. We define $N(C_1,C_2)$ as the number of occurrences of a tree of $C_2$ as child of $\roottree(C_1)$. The tree reduction $\DAG(\tau)$ is defined as the quotient graph $Q$ augmented with labels $N(C_1,C_2)$ on its edges (see \cite[Definition 3 (Reduction of a tree)]{GF2010} for more details). Intuitively, the labeled graph $\DAG(\tau)$ represents the original tree $\tau$ without its structural redundancies. Illustrations are presented in Figures~\ref{fig:1} and \ref{fig:1bis}. It should be noticed that a tree can be exactly reconstructed from its DAG reduction \cite[Proposition 4]{GF2010}, i.e., the application $\DAG$ is a one-to-one correspondence from unordered trees into DAG reductions space, which inverse is denoted $\tree$.

\begin{figure}[h]
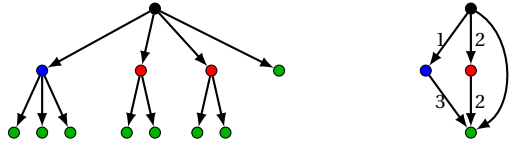

\centering
\treeTA
\caption{An unordered tree $\tau$ (left) and its reduction $\DAG(\tau)$ (right). In the tree, roots of isomorphic subtrees are colored identically. In the quotient graph, vertices are equivalence classes colored according to the class of isomorphic subtrees of $\tau$ that they represent.}
\label{fig:1}
\end{figure}

\begin{figure}[h]
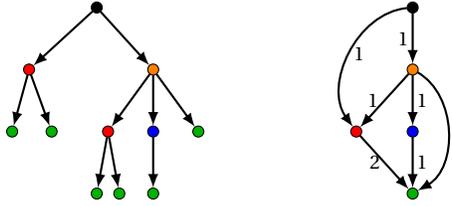

\centering
\treeTC
\caption{Another unordered tree $\tau$ (left) and its reduction $\DAG(\tau)$ (right).}
\label{fig:1bis}
\end{figure}

\paragraph{Notation of DAGs} The vertices of the quotient graph can be sorted by height, i.e., numbered as,
$$V_\equiv = \{ (h,i) ~:~0\leq h\leq \height(\tau),~1\leq i\leq M_h\},$$
where $M_h$ denotes the number of vertices of height $h$ appearing in $Q$. In other words, $(h,i)$ is the vertex representing the $i^\text{th}$ equivalence class of height $h$ appearing in $\tau$. We highlight that the order chosen to number the equivalence classes of a same height is not significant. The structure of $\DAG(\tau)$ is thus fully characterized by the array
$$\left[ N( (h_1,i) , (h_2,j)) \right]_{h_2 , h_1, i , j} ,$$
with $0\leq h_2<h_1\leq\height(\tau)$, $1\leq i\leq M_{h_1}$ and $1\leq j\leq M_{h_2}$, and where
$$N((h_1,i) , (h_2,j))=0\quad\Leftrightarrow\quad( (h_1,i) , (h_2,j))\notin E_\equiv .$$
If $(h,i)\in V_\equiv$ and $D=\DAG(\tau)$, $D[(h,i)]$ denotes the sub-DAG rooted at $(h,i)$. By construction of $D$, $D[(h,i)]$ is the reduction of a unique (up to an isomorphism) subtree of height $h$ appearing in $\tau$, which is denoted by $\tree(D[(h,i)])$.

\subsection{Self-nested trees}
\label{ss:sntrees}

A tree $\tau$ is called self-nested \cite[III. Self-nested trees]{GF2010} if, for any pair of vertices $v$ and $w$ such that $\height(\tau[v])=\height(\tau[w])$, $\tau[v]$ and $\tau[w]$ are isomorphic. It should be noted that this characterization is equivalent to the following statement: for any pair of vertices $v$ and $w$ such that $\height(\tau[v])\leq\height(\tau[w])$, $\tau[v]$ is (isomorphic to) a subtree of $\tau[w]$. By definition, self-nested trees achieve the maximal presence of redundancies in their structure. Self-nested trees are tightly connected with linear DAGs, i.e., DAGs containing at least one path that goes through all their vertices.

\begin{proposition}[Godin and Ferraro \cite{GF2010}]
\label{prop:equivalence}
A tree $\tau$ is self-nested if and only if its reduction $\DAG(\tau)$ is linear.
\end{proposition}

We point out that a linear DAG reduction means that the compression is optimal, at least in terms of number of vertices. Indeed, the number of vertices of $\DAG(\tau)$ is always greater than $\height(\tau)+1$ by construction: there is at least one tree of height $h$ appearing in $\tau$ for $0\leq h\leq\height(\tau)$. In addition, the inequality is saturated if and only if the DAG is linear. Proposition~\ref{prop:equivalence} evidences that DAG compression is optimal for self-nested trees. This is because (i) DAG compression removes repetitions of subtrees and (ii) self-nested trees present systematic redundancies in their subtrees. Two examples are displayed in Figures~\ref{fig:2} and \ref{fig:2bis} for illustration purposes.

\begin{figure}[h]
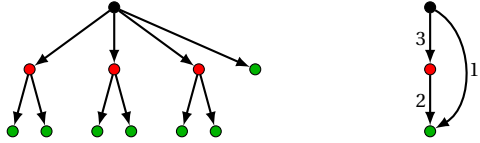

\centering
\treeTB
\caption{A self-nested tree $\tau$ (left) and its linear reduction $\mathcal{R}(\tau)$ (right). In the tree, all the subtrees of the same height are isomorphic and their roots are colored identically. The quotient graph is a linear DAG in which each vertex represents all the subtrees with the same height.}
\label{fig:2}
\end{figure}

\begin{figure}[h]
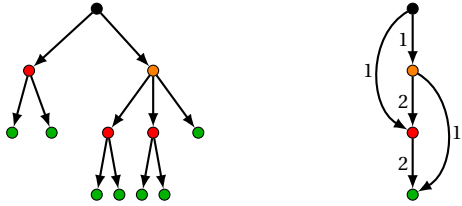

\centering
\treeTCsn
\caption{Another self-nested tree $\tau$ (left) and its linear reduction $\DAG(\tau)$ (right).}
\label{fig:2bis}
\end{figure}

\paragraph{Notation of linear DAGs} The vertices of a linear DAG $D=(V_\equiv,E_\equiv,N)$ of height $H$ can be sorted by height and thus numbered as $V_\equiv = \{0,\dots, H\}$, i.e., $h$ denotes the index of the unique equivalence class of height $h$ appearing in the associated self-nested tree. The structure of $D$ is fully characterized by the array
$$[ N(h_1,h_2) ]_{0\leq h_2<h_1\leq H},$$
where $N(h_1,h_2)=0$ if and only if $(h_1,h_2)\notin E_\equiv$. One can also notice that $N(h_1,h_1-1)\geq1$.

\paragraph{Simulation of random self-nested trees} As stated in Proposition~\ref{prop:equivalence}, there is a one-to-one corres\-pon\-dence between self-nested trees and linear DAGs. To simulate a random self-nested tree, we generate a random linear DAG as follows. Given height $H$ and maximal outdegree $d$, all the coefficients $N(h_1,h_2)$, $0\leq h_2<h_1\leq H$, are chosen under the uniform distribution with constraints
$$\sum_{h_2=0}^{h_1-1} N(h_1,h_2) \leq d $$
and $N(h_1,h_1-1)\geq1$, by rejection sampling. In this paper, we aim to compare the algorithmic efficiency of self-nested trees with respect to general trees. We have generated the datasets used in the numerical experiments as follows: first, generate a random general tree of given size, and then generate a random self-nested tree with the same height and outdegree.


\section{Efficiency of self-nested trees}
\label{s:snt}

As aforementioned in Subsection~\ref{ss:sntrees}, self-nested trees present the highest level of redundancy in their subtrees. Thus, one can expect good algorithmic properties from them. In this section we prove their computational efficiency through three questions: compression, evaluation of bottom-up functions, evaluation of edit distance.

\subsection{Compression rates}

Proposition \ref{prop:equivalence} proves that self-nested trees achieve optimal compression rates among trees of the same height whatever their number of vertices. However, this statement does not take into account the number of edges, neither the presence of labels on the edges of the DAG reduction. We have estimated the average disk size being occupied by a tree and its DAG reduction from the simulation of $40\,000$ random trees (a half being self-nested). The data have been stored by using the \verb+pickle+ module in \verb+Python3+. The results are presented in Figure~\ref{fig:space} and Table~\ref{tab:space}. The simulations show that the compression rate (defined as the ratio of the compressed size over the uncompressed size) is around $10\%$ for a random tree regardless of its size, while it is approximately $0.3\%$ for a self-nested tree.

\begin{figure}[h]
\centering
\includegraphics[width=8.5cm]{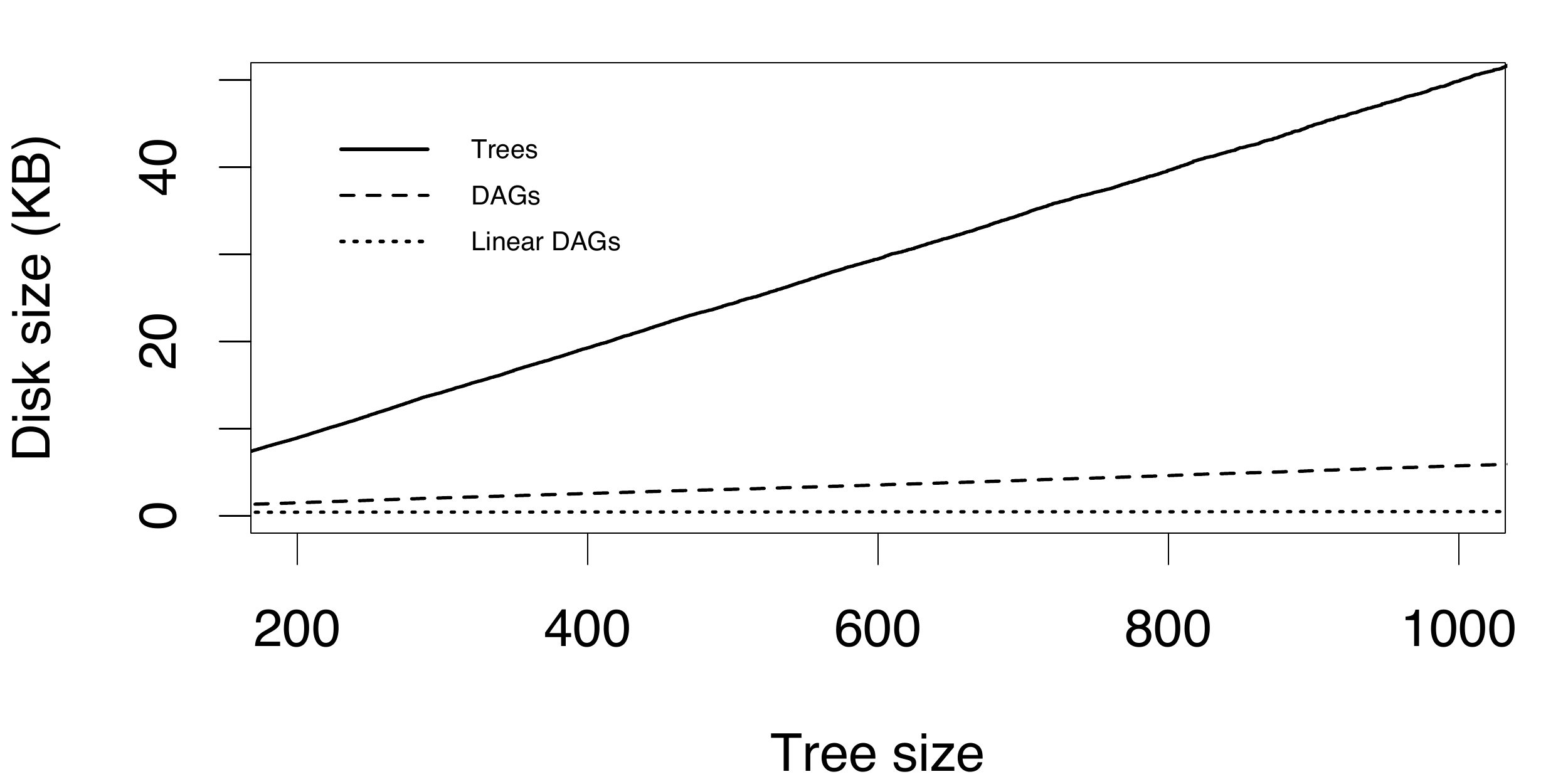}
\caption{Estimation of the disk space occupied by a tree and its DAG reduction for trees and self-nested trees.}
\label{fig:space}
\end{figure}

\begin{table}[h]
\centering
\begin{tabular}{lc}
&Disk size (KB)\\ \hline
Trees& $5.044\times 10^{-2}\,\#\tau$\\
DAGs& $5.433\times 10^{-3}\,\#\tau$ \\
Linear DAGs& $1.913\times10^{-4}\,\#\tau$\\ \hline
\end{tabular}
\caption{Estimation of the slope of the three curves of Figure~\ref{fig:space}. For example, the disk size being occupied by a tree of size $100$ is $5.044$ KB on average, while it is approximately $0.02$ KB for the DAG reduction of a self-nested tree of the same size.}
\label{tab:space}
\end{table}

\subsection{Bottom-up recursive functions}

Given a tree $\tau=(V,E)$, we consider bottom-up recursive functions $f : V \to\mathcal{X}$, defined by
$$
f(v) =
\left\{
\begin{array}{cl}
\Phi\big(\, (f(c))_{c\in\child(v)}\,\big) &\text{if}~\child(v)\neq\emptyset \\
f_0 &\text{else,}
\end{array}
\right.
$$
where $f_0$ is the value of $f$ on the leaves of $\tau$ and
$$\Phi:\bigcup_{n\geq1}^{}\mathcal{X}^n \to \mathcal{X} $$
is invariant under permutation of arguments. The value of $f(\tau)$ is defined as $f(\tau) = f(\roottree(\tau))$. Bottom-up recursive functions play a central role in the study of trees. For instance, the number of vertices, the number of leaves, the height and the Strahler number (that measures the branching complexity) are useful bottom-up recursive functions.

\begin{proposition}
\label{prop:bu:compl}
$f(\tau)$ can be computed in:
\begin{itemize}
\item $O(\#\tau)$-time from the tree $\tau$;
\item $O(\#D \degr(\tau))$-time from the DAG reduction $D=\DAG(\tau)$.
\end{itemize}
Consequently, if $\tau$ is self-nested, $f(\tau)$ can be computed in $O(\height(\tau) \degr(\tau))$-time.
\end{proposition}

For example, the number of vertices of a tree $\tau$ can be computed from its DAG reduction through the recursive formula
$$ \#\tree( D[(h_1,i)]) = 1+\sum_{h_2=0}^{h_1-1} \sum_{j=1}^{M_{h_2}} N( (h_1,i) , (h_2,j) ) \, \#\tree( D[(h_2,j] ) .$$
If $\tau$ is self-nested, its DAG reduction is linear and the recursion becomes
\begin{equation}
\label{eq:nnodes:sn:hp}
\#\tree(D[h_1]) = 1+\sum_{h_2=0}^{h_1-1} N(h_1,h_2) \,\#\tree(D[h_2]),
\end{equation}
where the number of positive terms in the sum is bounded by $\deg(\tau)$. The complexities stated in Proposition \ref{prop:bu:compl} have been illustrated through the numerical simulation of $40\,000$ random trees (see Figure~\ref{fig:comptime} and Table~\ref{tab:comptime}). The simulations state that, whatever the size of the tree, computing a bottom-up function is 4 times faster from a linear DAG than from the DAG reduction of a random tree, and nearly 20 times faster than from the tree.

\begin{figure}[h]
\centering
\includegraphics[width=8.5cm]{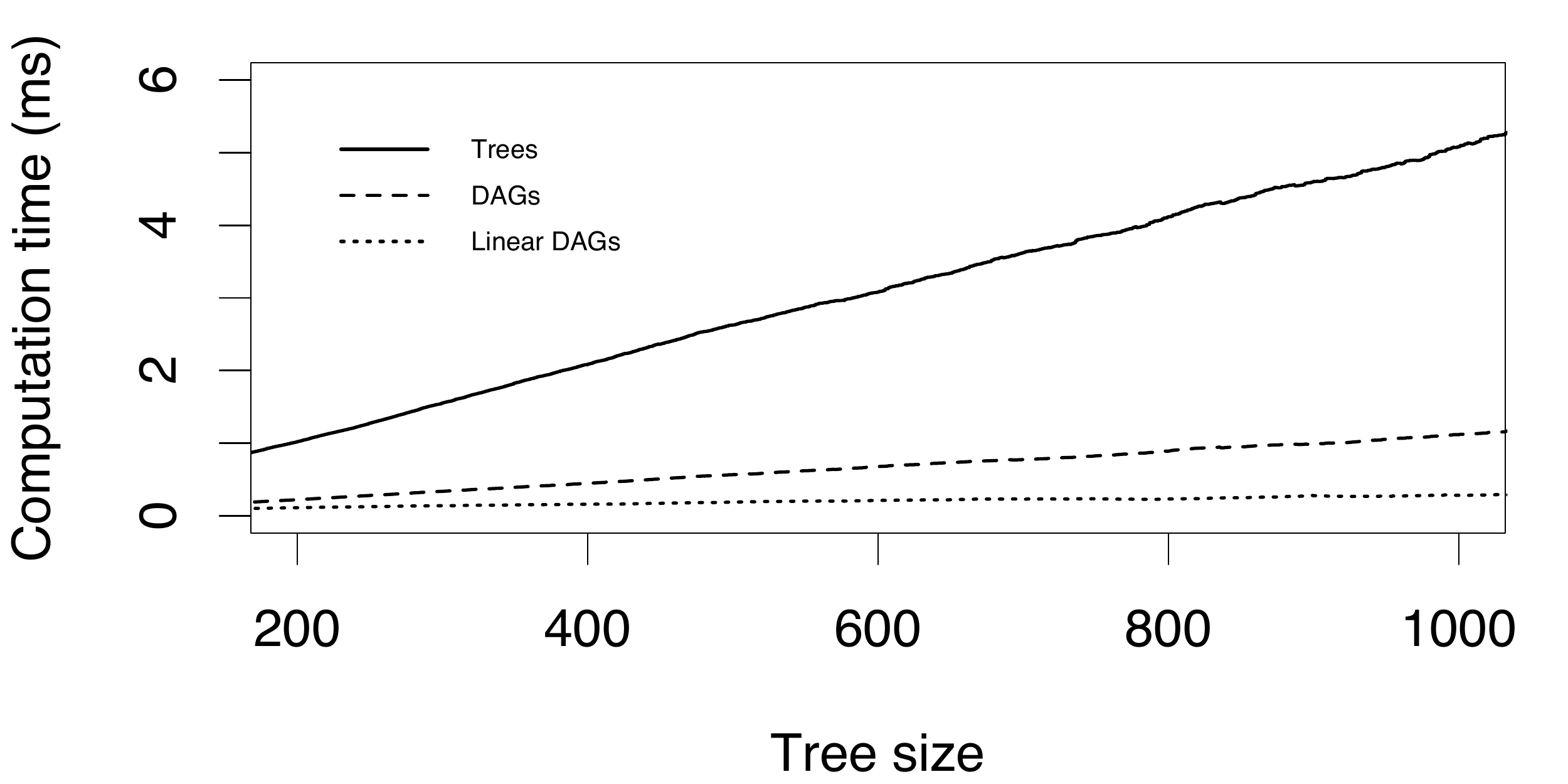}
\caption{Estimation of the computation time of the number of vertices from trees, DAGs and linear DAGs.}
\label{fig:comptime}
\end{figure}

\begin{table}[h]
\centering
\begin{tabular}{lc}
&Computation time (ms)\\ \hline
Trees& $5.094\times 10^{-3}\,\#\tau$\\
DAGs& $1.112\times 10^{-3}\,\#\tau$ \\
Linear DAGs& $2.652\times10^{-4}\,\#\tau$\\ \hline
\end{tabular}
\caption{Estimation of the slope of the three curves of Figure~\ref{fig:comptime}. For example, the computation time of the number of vertices from a tree of size $1\,000$ is $5.094$ ms on average, while it is approximately $0.27$ ms from the DAG reduction of a self-nested tree of the same size.}
\label{tab:comptime}
\end{table}

\subsection{Edit distance}

The problem of comparing trees occurs in several diverse areas such as computational biology and image analysis. We refer the reader to the survey \cite{B2005} in which the author reviews the available results and presents, in detail, one or more of the central algorithms for solving the problem. We consider a constrained edit distance between unordered rooted trees. This distance is based on the following tree edit operations \cite{xdiff}:
\begin{itemize}
\item\textit{Insertion.} Let $v$ be a vertex in a tree $\tau$. The insertion operation inserts a new vertex in the list of children of $v$. In the transformed tree, the new vertex is necessarily a leaf.
\item\textit{Deletion.} Let $l$ be a leaf vertex in a tree $\tau$. The deletion operation results in removing $l$ from $\tau$. That is, if $v$ is the parent of vertex $l$, the set of children of $v$ in the transformed tree is $\child(v)\setminus\{l\}$.
\end{itemize}
As in \cite{xdiff} for ordered trees, only adding and deleting a leaf vertex are allowed edit operations. An edit script is an ordered sequence of edit operations. The result of applying an edit script $s$ to a tree $\tau$ is the tree $\tau^s$ obtained by applying the component edit operations to $\tau$, in the order they appear in the script. The cost of an edit script $s$ is the number of edit operations $\#s$. In other words, we assign a unit cost to both allowed operations. Finally, given two unordered rooted trees $\tau_1$ and $\tau_2$, the constrained edit cost $\delta(\tau_1,\tau_2)$ is the length of the minimum edit script that transforms $\tau_1$ to a tree that is isomorphic to $\tau_2$,
$$\delta(\tau_1,\tau_2) = \min_{\{s\,:\,\tau_1^s \equiv \tau_2\}} \#s.$$
We show in Lemma \ref{prop:delta:dist} in the supplementary document that $\delta$ defines a distance on the space of unordered trees.

\begin{proposition}
\label{prop:delta:complexity}
The edit distance $\delta(\tau_1,\tau_2)$ can be computed in:
\begin{itemize}
\item $O(\#\tau_1 \, \#\tau_2 \, \psi(\tau_1,\tau_2))$-time from the trees $\tau_1$ and $\tau_2$;
\item $O(\#D_1 \#D_2 \, \degr(D_1)\degr(D_2) \, \psi(\tau_1,\tau_2))$-time from the DAG reductions $D_1=\DAG(\tau_1)$ and $D_2=\DAG(\tau_1)$;
\end{itemize}
with
$$\psi(\tau_1,\tau_2) = (\degr\,\tau_1+\degr\,\tau_2) \log_2(\degr\,\tau_1+\degr\,\tau_2) .$$
Consequently, if $\tau_1$ and $\tau_2$ are self-nested, the time-complexity of $\delta(\tau_1,\tau_2)$ is
$$O(\height(\tau_1) \height(\tau_2) \, \degr(D_1) \degr(D_2)\, \psi(\tau_1,\tau_2)).$$
\end{proposition}

The complexities stated in Proposition \ref{prop:bu:compl} have been illustrated through the numerical simulation of 20\,000 pairs of random trees (see Figure~\ref{fig:comptimedist}). Again, the simulations show that this complex query can be answered much faster for self-nested trees than for general trees. We highlight that the computational gain is even more substantial for this quadratic operation than for computing bottom-up functions.

\begin{figure}[h]
\centering
\includegraphics[width=8.5cm]{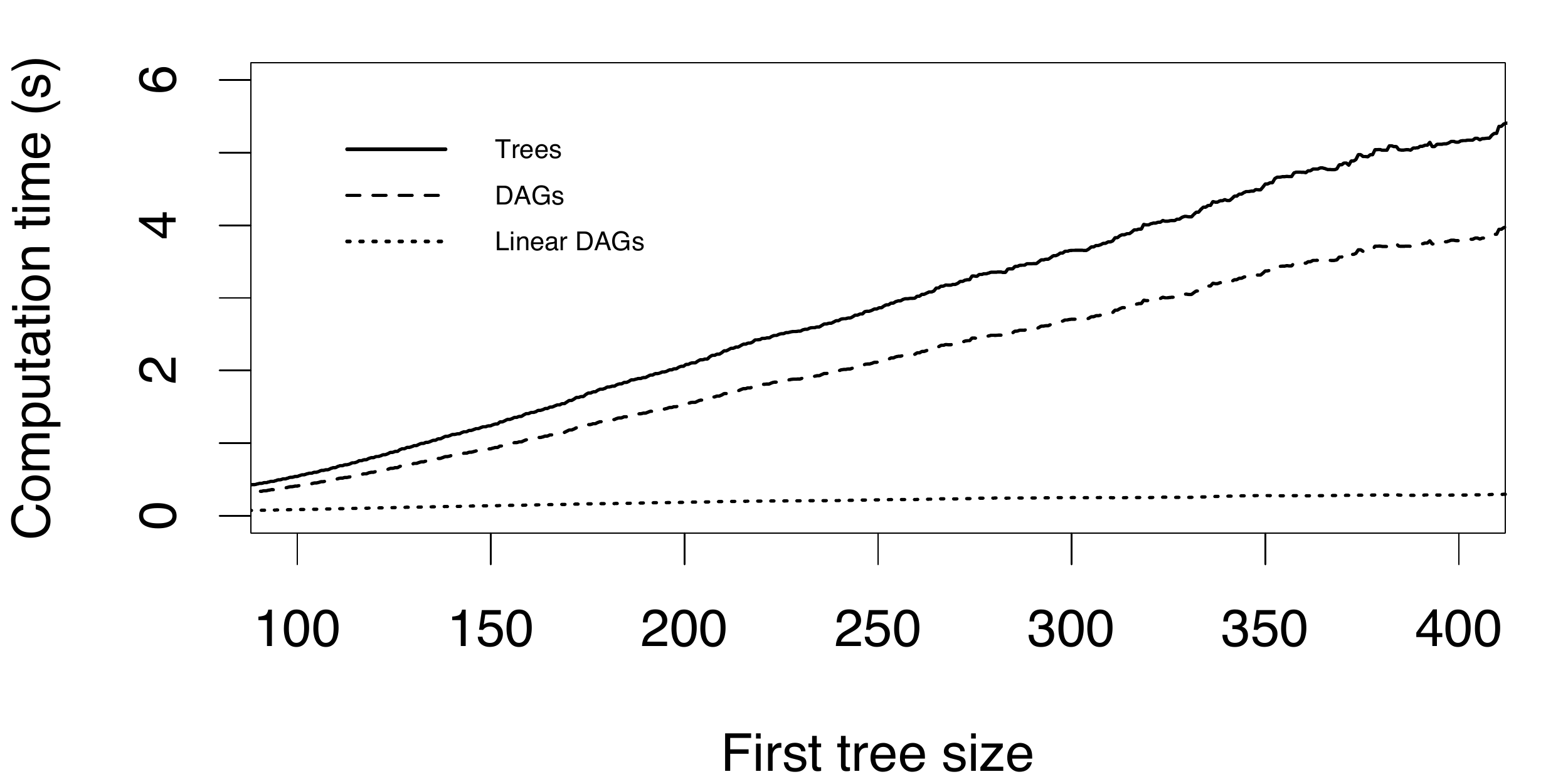}
\caption{Estimation of the computation time of the edit distance between two trees from the trees and from their DAG reduction, for trees and self-nested trees. The computation time is displayed as a function of the first tree size since the edit distance is a symmetric function.}
\label{fig:comptimedist}
\end{figure}


\section{Combinatorics of self-nested trees}
\label{section:properties}

We now investigate combinatorics of self-nested trees. This section gathers new results about this problem for trees that satisfy constraints on the height and the outdegree. In this context, $\mathbb{T}_{=H,\leq d}$ ($\mathbb{T}_{\leq H,\leq d}$, respectively) denotes the set of unordered trees of height $H$ (of height bounded by $H$, respectively) and outdegree bounded by $d$. The same notation lies for self-nested trees with the exponent $sn$. We give an explicit formula for the cardinality of self-nested trees in the following proposition.

\begin{proposition}\label{prop1}
For any $H\geq1$ and $d\geq1$,
$$\#\mathbb{T}^{sn}_{=H,\leq d} = \prod_{i=1}^{H} {\binom{d+H-i}{H-i+1}} .$$
\end{proposition}

A more traditional approach in the literature is to investigate combinatorics of trees with a given number of vertices. For example, exploiting the theory of ordinary generating functions, Flajolet and Sedgewick recursively obtained the cardinality of the set $\mathbb{T}_{n}$ of unordered trees with $n$ vertices (see \cite[eq. (73)]{FS2009} and OEIS\,\footnote{On-line Encyclopedia of Integer Sequences} A000081). In particular, the generating function associated with the unordered trees is given by
$$H(z) = z+z^2+2z^3+4z^4+9z^5+20z^6+48z^7+115 z^8+\cdots ,$$
where the coefficient $H_n$ of $z^n$ in $H(z)$ is the cardinality of the set $\mathbb{T}_{n}$. It would be very interesting to investigate the cardinality of the set $\mathbb{T}^{sn}_n$ of self-nested trees of size $n$. A strategy could be to remark that
$$
\#\mathbb{T}_n^{sn} = \sum_{H=1}^n \#\{\tau\in\mathbb{T}_n^{sn}~:~\height(\tau) = H\} ,
$$
where $\#\{\tau\in\mathbb{T}_n^{sn}~:~\height(\tau) = H\}$ is a polynomial equation of degree $H$ in $H(H+1)/2$ unknown variables in light of \eqref{eq:nnodes:sn:hp}. Determining the number of solutions of such a Diophantine equation, even in this particular framework, remains a very difficult question.

Nevertheless, thanks to Proposition~\ref{prop1}, we can numerically evaluate the frequency of self-nested trees (see Table~\ref{tableau1}). We have also derived an asymptotic equivalent when both the height and the outdegree go to infinity that can be compared to the cardinality of unordered trees.

\begin{table}[h]
\centering
\begin{tabular}{cc|ccc}
									&&\multicolumn{3}{c}{outdegree}												\\
									& 				& $\leq2$				& $\leq3$			& $\leq4$				\\ \hline
\multirow{4}{*}{\rotatebox{90}{height}}		&	$\leq2$		& $0.88$				& $6.18.10^{-1}$	& $3.52.10^{-1}$		\\
									&	$\leq3$		& $0.49$				& $3.38.10^{-2}$	& $7.43.10^{-5}$		\\
									&	$\leq4$		& $0.07$				& $2.90.10^{-8}$	& $4.16.10^{-23}$		\\
									&	$\leq5$		& $3.36.10^{-4}$		& $3.56.10^{-28}$	& $1.66.10^{-100}$		\\
\end{tabular}
\caption{Relative frequencies of self-nested trees with given maximal height and ramification number within the set of unordered trees under the same constraint.}
\label{tableau1}
\end{table}
\begin{proposition}\label{prop2}
When $H$ and $d$ simultaneously go to infinity,
$$ \log\,\#\mathbb{T}_{=H,\leq d}^{sn} \sim \frac{(d+H)^2}{2}\log(d+H)- \frac{H^2}{2}\log\,H - \frac{d^2}{2}\log\,d-H d \log\,d .$$
For the sake of comparison,
$$\log\,\#\mathbb{T}_{\leq H,\leq d} = \Theta(d^{H-1}).$$
\end{proposition}

Consequently, self-nested trees are very rare among unordered trees, since they are roughly exponentially less frequent. Exploring the space of self-nested trees is thus exponentially easier than for general trees. This is a very good point in NP problems for which exhaustive search is often chosen to determine a solution.


\section{Approximation algorithm}
\label{s:approx}

In Section~\ref{s:snt}, we have established the numerical efficiency of self-nested trees. In addition, they are very unfrequent as seen in Proposition~\ref{prop2}, which can be an interesting asset in exhaustive search problems. Our aim is to develop an approximation algorithm of trees by self-nested trees in order to take advantage of these remarkable algorithmic properties for any unordered tree. An application of this algorithm will be presented in Section~\ref{s:fastpred}.

\subsection{Characterization of self-nested trees}

Let $\tau$ be a tree of height $H$ and $D=\DAG(\tau)$ its DAG reduction. With the notation of Subsection \ref{ss:treereduction}, for each vertex $(h,i)$ of $D$, we define
$$\nu(h_1,i,h_2) = \sum_{j=1}^{M_{h_2}} N ( \, (h_1,i) \, , \, (h_2,j)\,)$$
that counts the number of subtrees of height $h_2$ which root is a child of $(h_1,i)$. This quantity typically represents the height profile of the tree $\tau$: it gives the distribution of the number of subtrees of height $h_2$ under vertices of height $h_1$ in $\tau$. This feature makes us able to derive a new characterization of self-nested trees.
\begin{proposition}
$\tau$ is self-nested if and only if, for any $0\leq h_2 < h_1\leq H$ and $1\leq i,j\leq M_{h_1}$,
\begin{equation}\label{eq:charac}\nu(h_1,i,h_2) = \nu(h_1,j,h_2) .\end{equation}
\end{proposition}

In other words, a tree is self-nested if and only if its height profile is reduced to an array of Dirac masses. Furthermore, it should be remarked that, if \eqref{eq:charac} holds, the self-nested tree $\tau$ can be reconstructed from the array
$$[\nu(h_1,1,h_2)]_{0\leq h_2<h_1\leq H}.$$
Indeed, with the notation of Subsection~\ref{ss:sntrees} for self-nested trees, the label on edge $(h_1,h_2)$ in $D$ is $N(h_1,h_2) = \nu(h_1,1,h_2)$.

\subsection{Averaging}

Given a tree $\tau$, we construct the self-nested tree $\widehat{\tau}$ that op\-ti\-mally approxi\-ma\-tes, for each possible value of $(h_1,h_2)$, the quantity $\nu(h_1,\cdot,h_2)$ in weighted $\mathbb{L}^2$-norm taking into account the multiplicity of the vertices $(h_1,i)$ of $D$, $1\leq i\leq M_{h_1}$. The multiplicity $\mu((h_1,i))$ of a vertex $(h_1,i)$ of the DAG reduction $D$ of $\tau$ is the number of occurrences of the tree $\tree(D[(h_1,i)])$ in $\tau$. It is easy to see that $\mu((h_1,i))$ can be recursively computed from $D$ as
\begin{equation}
\mu( (h_1,i) ) = \prod_{{h>h_1} \atop {1\leq j\leq M_h}} N( (h,j) , (h_1,i) )\times\mu( (h,j) ) .\label{eq:multipli}
\end{equation}
The linear DAG $\widehat{D}$ of $\widehat{\tau}$ is defined (with the notation of Subsection~\ref{ss:sntrees}) by, for any $0\leq h_2<h_1\leq\height(\tau)$,
\begin{eqnarray*}
\widehat{N}(h_1,h_2) 
&=&\argmin_{x\in\mathbb{N}} \sum_{1\leq i\leq M_{h_1}} \mu( (h_1,i) )\,\big[x -\nu(h_1,i,h_2) \big]^2\\
&= &\pi\left( \frac{\displaystyle\sum_{1\leq i\leq M_{h_1}} \mu((h_1,i)) \sum_{1\leq j \leq M_{h_2}} N( (h_1,i) , (h_2,j) )}{\displaystyle\sum_{1\leq i\leq M_{h_1} } \mu( (h_1,i) )} \right) ,
\end{eqnarray*}
where $\pi$ denotes the projection on $\mathbb{N}$, i.e., the distribution of the number of subtrees of height $h_2$ under vertices of height $h_1$ in $\tau$ is approximated by its weighted mean. By construction, $\widehat{\tau}$ is the best self-nested approximation of the height profile of $\tau$. An example is presented in Figure~\ref{fig:exav}. The complexity of this algorithm is stated in Proposition~\ref{prop:complav}.

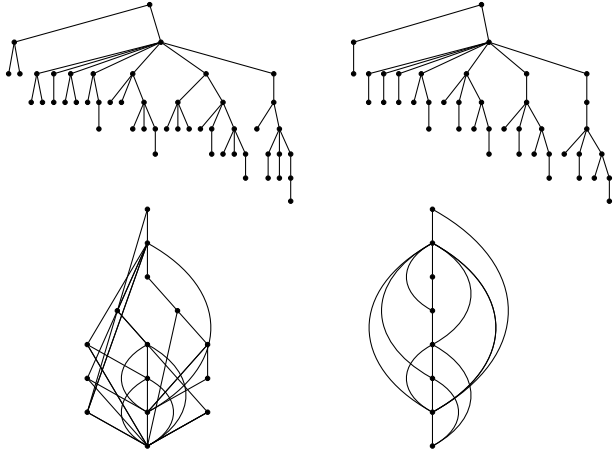
\begin{figure}[h!]
\centering
\def\xscale{0.15}
\def\yscale{0.5}
\def\nodescale{0.2}
\begin{tikzpicture}[xscale=\xscale,yscale=\yscale]
\tikzstyle{fleche}=[-,>=latex]
\tikzstyle{noeudblack}=[draw,circle,fill=black,scale=\nodescale]
\def\localnodescalea{1}
\node[noeudblack,scale=\localnodescalea] (4545410832) at ({0},{0}) {};
\node[noeudblack,scale=\localnodescalea] (4545536072) at ({-12.0},{-1.0}) {};
\node[noeudblack,scale=\localnodescalea] (4545536128) at ({-12.5},{-1.8408964152537146}) {};
\node[noeudblack,scale=\localnodescalea] (4545536184) at ({-11.5},{-1.8408964152537146}) {};
\node[noeudblack,scale=\localnodescalea] (4545409656) at ({1.0},{-1.0}) {};
\node[noeudblack,scale=\localnodescalea] (4545505544) at ({-10.0},{-1.8408964152537146}) {};
\node[noeudblack,scale=\localnodescalea] (4545505488) at ({-10.5},{-2.600732100905307}) {};
\node[noeudblack,scale=\localnodescalea] (4545505432) at ({-9.5},{-2.600732100905307}) {};
\node[noeudblack,scale=\localnodescalea] (4545505152) at ({-8.5},{-1.8408964152537146}) {};
\node[noeudblack,scale=\localnodescalea] (4545507112) at ({-8.5},{-2.600732100905307}) {};
\node[noeudblack,scale=\localnodescalea] (4545507168) at ({-7.0},{-1.8408964152537146}) {};
\node[noeudblack,scale=\localnodescalea] (4545507224) at ({-7.5},{-2.600732100905307}) {};
\node[noeudblack,scale=\localnodescalea] (4545507280) at ({-6.5},{-2.600732100905307}) {};
\node[noeudblack,scale=\localnodescalea] (4545505376) at ({-5.0},{-1.8408964152537146}) {};
\node[noeudblack,scale=\localnodescalea] (4545505320) at ({-5.5},{-2.600732100905307}) {};
\node[noeudblack,scale=\localnodescalea] (4545505264) at ({-4.5},{-2.600732100905307}) {};
\node[noeudblack,scale=\localnodescalea] (4545505208) at ({-4.5},{-3.3078388820918545}) {};
\node[noeudblack,scale=\localnodescalea] (4545505992) at ({-1.5},{-1.8408964152537146}) {};
\node[noeudblack,scale=\localnodescalea] (4545505936) at ({-3.5},{-2.600732100905307}) {};
\node[noeudblack,scale=\localnodescalea] (4545505600) at ({-2.5},{-2.600732100905307}) {};
\node[noeudblack,scale=\localnodescalea] (4545505880) at ({-0.5},{-2.600732100905307}) {};
\node[noeudblack,scale=\localnodescalea] (4545505824) at ({-1.5},{-3.3078388820918545}) {};
\node[noeudblack,scale=\localnodescalea] (4545505768) at ({-0.5},{-3.3078388820918545}) {};
\node[noeudblack,scale=\localnodescalea] (4545505712) at ({0.5},{-3.3078388820918545}) {};
\node[noeudblack,scale=\localnodescalea] (4545505656) at ({0.5},{-3.9765791870682765}) {};
\node[noeudblack,scale=\localnodescalea] (4545507000) at ({5.0},{-1.8408964152537146}) {};
\node[noeudblack,scale=\localnodescalea] (4545506440) at ({2.5},{-2.600732100905307}) {};
\node[noeudblack,scale=\localnodescalea] (4545503416) at ({1.5},{-3.3078388820918545}) {};
\node[noeudblack,scale=\localnodescalea] (4545504368) at ({2.5},{-3.3078388820918545}) {};
\node[noeudblack,scale=\localnodescalea] (4545506608) at ({3.5},{-3.3078388820918545}) {};
\node[noeudblack,scale=\localnodescalea] (4545506496) at ({6.5},{-2.600732100905307}) {};
\node[noeudblack,scale=\localnodescalea] (4545506048) at ({4.5},{-3.3078388820918545}) {};
\node[noeudblack,scale=\localnodescalea] (4545506160) at ({5.5},{-3.3078388820918545}) {};
\node[noeudblack,scale=\localnodescalea] (4545506104) at ({5.5},{-3.9765791870682765}) {};
\node[noeudblack,scale=\localnodescalea] (4545506552) at ({7.5},{-3.3078388820918545}) {};
\node[noeudblack,scale=\localnodescalea] (4545506272) at ({6.5},{-3.9765791870682765}) {};
\node[noeudblack,scale=\localnodescalea] (4545506216) at ({7.5},{-3.9765791870682765}) {};
\node[noeudblack,scale=\localnodescalea] (4545506384) at ({8.5},{-3.9765791870682765}) {};
\node[noeudblack,scale=\localnodescalea] (4545506328) at ({8.5},{-4.615522291314549}) {};
\node[noeudblack,scale=\localnodescalea] (4545411448) at ({11.0},{-1.8408964152537146}) {};
\node[noeudblack,scale=\localnodescalea] (4545410328) at ({11.0},{-2.600732100905307}) {};
\node[noeudblack,scale=\localnodescalea] (4545506888) at ({9.5},{-3.3078388820918545}) {};
\node[noeudblack,scale=\localnodescalea] (4545503696) at ({11.5},{-3.3078388820918545}) {};
\node[noeudblack,scale=\localnodescalea] (4545503584) at ({10.5},{-3.9765791870682765}) {};
\node[noeudblack,scale=\localnodescalea] (4545503304) at ({10.5},{-4.615522291314549}) {};
\node[noeudblack,scale=\localnodescalea] (4545506720) at ({11.5},{-3.9765791870682765}) {};
\node[noeudblack,scale=\localnodescalea] (4545503640) at ({11.5},{-4.615522291314549}) {};
\node[noeudblack,scale=\localnodescalea] (4545507056) at ({12.5},{-3.9765791870682765}) {};
\node[noeudblack,scale=\localnodescalea] (4545506664) at ({12.5},{-4.615522291314549}) {};
\node[noeudblack,scale=\localnodescalea] (4545503360) at ({12.5},{-5.230310444265813}) {};
\draw[fleche] (4545410832)--(4545536072) {};
\draw[fleche] (4545410832)--(4545409656) {};
\draw[fleche] (4545536072)--(4545536128) {};
\draw[fleche] (4545536072)--(4545536184) {};
\draw[fleche] (4545409656)--(4545505544) {};
\draw[fleche] (4545409656)--(4545505152) {};
\draw[fleche] (4545409656)--(4545507168) {};
\draw[fleche] (4545409656)--(4545505376) {};
\draw[fleche] (4545409656)--(4545505992) {};
\draw[fleche] (4545409656)--(4545507000) {};
\draw[fleche] (4545409656)--(4545411448) {};
\draw[fleche] (4545505544)--(4545505488) {};
\draw[fleche] (4545505544)--(4545505432) {};
\draw[fleche] (4545505152)--(4545507112) {};
\draw[fleche] (4545507168)--(4545507224) {};
\draw[fleche] (4545507168)--(4545507280) {};
\draw[fleche] (4545505376)--(4545505320) {};
\draw[fleche] (4545505376)--(4545505264) {};
\draw[fleche] (4545505264)--(4545505208) {};
\draw[fleche] (4545505992)--(4545505936) {};
\draw[fleche] (4545505992)--(4545505600) {};
\draw[fleche] (4545505992)--(4545505880) {};
\draw[fleche] (4545505880)--(4545505824) {};
\draw[fleche] (4545505880)--(4545505768) {};
\draw[fleche] (4545505880)--(4545505712) {};
\draw[fleche] (4545505712)--(4545505656) {};
\draw[fleche] (4545507000)--(4545506440) {};
\draw[fleche] (4545507000)--(4545506496) {};
\draw[fleche] (4545506440)--(4545503416) {};
\draw[fleche] (4545506440)--(4545504368) {};
\draw[fleche] (4545506440)--(4545506608) {};
\draw[fleche] (4545506496)--(4545506048) {};
\draw[fleche] (4545506496)--(4545506160) {};
\draw[fleche] (4545506496)--(4545506552) {};
\draw[fleche] (4545506160)--(4545506104) {};
\draw[fleche] (4545506552)--(4545506272) {};
\draw[fleche] (4545506552)--(4545506216) {};
\draw[fleche] (4545506552)--(4545506384) {};
\draw[fleche] (4545506384)--(4545506328) {};
\draw[fleche] (4545411448)--(4545410328) {};
\draw[fleche] (4545410328)--(4545506888) {};
\draw[fleche] (4545410328)--(4545503696) {};
\draw[fleche] (4545503696)--(4545503584) {};
\draw[fleche] (4545503696)--(4545506720) {};
\draw[fleche] (4545503696)--(4545507056) {};
\draw[fleche] (4545503584)--(4545503304) {};
\draw[fleche] (4545506720)--(4545503640) {};
\draw[fleche] (4545507056)--(4545506664) {};
\draw[fleche] (4545506664)--(4545503360) {};
\end{tikzpicture}
\qquad
\def\xscale{0.2}
\begin{tikzpicture}[xscale=\xscale,yscale=\yscale]
\tikzstyle{fleche}=[-,>=latex]
\tikzstyle{noeudblack}=[draw,circle,fill=black,scale=\nodescale]
\def\localnodescalea{1}
\node[noeudblack,scale=\localnodescalea] (4545538088) at ({0},{0}) {};
\node[noeudblack,scale=\localnodescalea] (4545538704) at ({-8.5},{-1.0}) {};
\node[noeudblack,scale=\localnodescalea] (4545538536) at ({-8.5},{-1.8408964152537146}) {};
\node[noeudblack,scale=\localnodescalea] (4545538144) at ({0.5},{-1.0}) {};
\node[noeudblack,scale=\localnodescalea] (4545539096) at ({-7.5},{-1.8408964152537146}) {};
\node[noeudblack,scale=\localnodescalea] (4545539208) at ({-7.5},{-2.600732100905307}) {};
\node[noeudblack,scale=\localnodescalea] (4545538872) at ({-6.5},{-1.8408964152537146}) {};
\node[noeudblack,scale=\localnodescalea] (4545539264) at ({-6.5},{-2.600732100905307}) {};
\node[noeudblack,scale=\localnodescalea] (4545539040) at ({-5.5},{-1.8408964152537146}) {};
\node[noeudblack,scale=\localnodescalea] (4545539488) at ({-5.5},{-2.600732100905307}) {};
\node[noeudblack,scale=\localnodescalea] (4545538928) at ({-4.0},{-1.8408964152537146}) {};
\node[noeudblack,scale=\localnodescalea] (4545538424) at ({-4.5},{-2.600732100905307}) {};
\node[noeudblack,scale=\localnodescalea] (4545539320) at ({-3.5},{-2.600732100905307}) {};
\node[noeudblack,scale=\localnodescalea] (4545539376) at ({-3.5},{-3.3078388820918545}) {};
\node[noeudblack,scale=\localnodescalea] (4545538760) at ({-1.0},{-1.8408964152537146}) {};
\node[noeudblack,scale=\localnodescalea] (4545539544) at ({-2.5},{-2.600732100905307}) {};
\node[noeudblack,scale=\localnodescalea] (4545539600) at ({-1.5},{-2.600732100905307}) {};
\node[noeudblack,scale=\localnodescalea] (4545539712) at ({-1.5},{-3.3078388820918545}) {};
\node[noeudblack,scale=\localnodescalea] (4545539656) at ({0.0},{-2.600732100905307}) {};
\node[noeudblack,scale=\localnodescalea] (4545539824) at ({-0.5},{-3.3078388820918545}) {};
\node[noeudblack,scale=\localnodescalea] (4545539880) at ({0.5},{-3.3078388820918545}) {};
\node[noeudblack,scale=\localnodescalea] (4545539992) at ({0.5},{-3.9765791870682765}) {};
\node[noeudblack,scale=\localnodescalea] (4545539432) at ({3.0},{-1.8408964152537146}) {};
\node[noeudblack,scale=\localnodescalea] (4545539936) at ({3.0},{-2.600732100905307}) {};
\node[noeudblack,scale=\localnodescalea] (4545568840) at ({1.5},{-3.3078388820918545}) {};
\node[noeudblack,scale=\localnodescalea] (4545568896) at ({2.5},{-3.3078388820918545}) {};
\node[noeudblack,scale=\localnodescalea] (4545569008) at ({2.5},{-3.9765791870682765}) {};
\node[noeudblack,scale=\localnodescalea] (4545568952) at ({4.0},{-3.3078388820918545}) {};
\node[noeudblack,scale=\localnodescalea] (4545569120) at ({3.5},{-3.9765791870682765}) {};
\node[noeudblack,scale=\localnodescalea] (4545569176) at ({4.5},{-3.9765791870682765}) {};
\node[noeudblack,scale=\localnodescalea] (4545569288) at ({4.5},{-4.615522291314549}) {};
\node[noeudblack,scale=\localnodescalea] (4545539768) at ({7.0},{-1.8408964152537146}) {};
\node[noeudblack,scale=\localnodescalea] (4545569064) at ({7.0},{-2.600732100905307}) {};
\node[noeudblack,scale=\localnodescalea] (4545569344) at ({7.0},{-3.3078388820918545}) {};
\node[noeudblack,scale=\localnodescalea] (4545569456) at ({5.5},{-3.9765791870682765}) {};
\node[noeudblack,scale=\localnodescalea] (4545569512) at ({6.5},{-3.9765791870682765}) {};
\node[noeudblack,scale=\localnodescalea] (4545569624) at ({6.5},{-4.615522291314549}) {};
\node[noeudblack,scale=\localnodescalea] (4545569568) at ({8.0},{-3.9765791870682765}) {};
\node[noeudblack,scale=\localnodescalea] (4545569736) at ({7.5},{-4.615522291314549}) {};
\node[noeudblack,scale=\localnodescalea] (4545569792) at ({8.5},{-4.615522291314549}) {};
\node[noeudblack,scale=\localnodescalea] (4545569904) at ({8.5},{-5.230310444265813}) {};
\draw[fleche] (4545538088)--(4545538704) {};
\draw[fleche] (4545538088)--(4545538144) {};
\draw[fleche] (4545538704)--(4545538536) {};
\draw[fleche] (4545538144)--(4545539096) {};
\draw[fleche] (4545538144)--(4545538872) {};
\draw[fleche] (4545538144)--(4545539040) {};
\draw[fleche] (4545538144)--(4545538928) {};
\draw[fleche] (4545538144)--(4545538760) {};
\draw[fleche] (4545538144)--(4545539432) {};
\draw[fleche] (4545538144)--(4545539768) {};
\draw[fleche] (4545539096)--(4545539208) {};
\draw[fleche] (4545538872)--(4545539264) {};
\draw[fleche] (4545539040)--(4545539488) {};
\draw[fleche] (4545538928)--(4545538424) {};
\draw[fleche] (4545538928)--(4545539320) {};
\draw[fleche] (4545539320)--(4545539376) {};
\draw[fleche] (4545538760)--(4545539544) {};
\draw[fleche] (4545538760)--(4545539600) {};
\draw[fleche] (4545538760)--(4545539656) {};
\draw[fleche] (4545539600)--(4545539712) {};
\draw[fleche] (4545539656)--(4545539824) {};
\draw[fleche] (4545539656)--(4545539880) {};
\draw[fleche] (4545539880)--(4545539992) {};
\draw[fleche] (4545539432)--(4545539936) {};
\draw[fleche] (4545539936)--(4545568840) {};
\draw[fleche] (4545539936)--(4545568896) {};
\draw[fleche] (4545539936)--(4545568952) {};
\draw[fleche] (4545568896)--(4545569008) {};
\draw[fleche] (4545568952)--(4545569120) {};
\draw[fleche] (4545568952)--(4545569176) {};
\draw[fleche] (4545569176)--(4545569288) {};
\draw[fleche] (4545539768)--(4545569064) {};
\draw[fleche] (4545569064)--(4545569344) {};
\draw[fleche] (4545569344)--(4545569456) {};
\draw[fleche] (4545569344)--(4545569512) {};
\draw[fleche] (4545569344)--(4545569568) {};
\draw[fleche] (4545569512)--(4545569624) {};
\draw[fleche] (4545569568)--(4545569736) {};
\draw[fleche] (4545569568)--(4545569792) {};
\draw[fleche] (4545569792)--(4545569904) {};
\end{tikzpicture}\\
\def\xscale{0.8}
\def\yscale{0.45}
\begin{tikzpicture}[xscale=\xscale,yscale=\yscale]
\tikzstyle{fleche}=[-,>=latex]
\tikzstyle{noeud}=[draw,circle,fill=black,scale=\nodescale*1]
\node[noeud] (0) at ({-0.5},{0}) {};
\tikzstyle{noeud}=[draw,circle,fill=black,scale=\nodescale*1]
\node[noeud] (1) at ({-1.5},{1}) {};
\tikzstyle{noeud}=[draw,circle,fill=black,scale=\nodescale*1]
\node[noeud] (2) at ({-0.5},{1}) {};
\tikzstyle{noeud}=[draw,circle,fill=black,scale=\nodescale*1]
\node[noeud] (3) at ({0.5},{1}) {};
\tikzstyle{noeud}=[draw,circle,fill=black,scale=\nodescale*1]
\node[noeud] (4) at ({-1.5},{2}) {};
\tikzstyle{noeud}=[draw,circle,fill=black,scale=\nodescale*1]
\node[noeud] (5) at ({-0.5},{2}) {};
\tikzstyle{noeud}=[draw,circle,fill=black,scale=\nodescale*1]
\node[noeud] (6) at ({0.5},{2}) {};
\tikzstyle{noeud}=[draw,circle,fill=black,scale=\nodescale*1]
\node[noeud] (7) at ({-1.5},{3}) {};
\tikzstyle{noeud}=[draw,circle,fill=black,scale=\nodescale*1]
\node[noeud] (8) at ({-0.5},{3}) {};
\tikzstyle{noeud}=[draw,circle,fill=black,scale=\nodescale*1]
\node[noeud] (9) at ({0.5},{3}) {};
\tikzstyle{noeud}=[draw,circle,fill=black,scale=\nodescale*1]
\node[noeud] (10) at ({-1.0},{4}) {};
\tikzstyle{noeud}=[draw,circle,fill=black,scale=\nodescale*1]
\node[noeud] (11) at ({0.0},{4}) {};
\tikzstyle{noeud}=[draw,circle,fill=black,scale=\nodescale*1]
\node[noeud] (12) at ({-0.5},{5}) {};
\tikzstyle{noeud}=[draw,circle,fill=black,scale=\nodescale*1]
\node[noeud] (13) at ({-0.5},{6}) {};
\tikzstyle{noeud}=[draw,circle,fill=black,scale=\nodescale*1]
\node[noeud] (14) at ({-0.5},{7}) {};
\draw[fleche] (1)--(0) {};
\draw[fleche] (1)--(0) {};
\draw[fleche] (2)--(0) {};
\draw[fleche] (3)--(0) {};
\draw[fleche] (3)--(0) {};
\draw[fleche] (3)--(0) {};
\draw[fleche] (4)--(0) {};
\draw[fleche] (4)--(2) {};
\draw[fleche] (5) to [bend left=45](0) {};
\draw[fleche] (5) to [bend right=45](0) {};
\draw[fleche] (5)--(2) {};
\draw[fleche] (6)--(2) {};
\draw[fleche] (7)--(0) {};
\draw[fleche] (7)--(0) {};
\draw[fleche] (7)--(5) {};
\draw[fleche] (8) to [bend left=45](0) {};
\draw[fleche] (8) to [bend right=45](2) {};
\draw[fleche] (8)--(5) {};
\draw[fleche] (9)--(2) {};
\draw[fleche] (9)--(2) {};
\draw[fleche] (9)--(6) {};
\draw[fleche] (10)--(3) {};
\draw[fleche] (10)--(8) {};
\draw[fleche] (11)--(0) {};
\draw[fleche] (11)--(9) {};
\draw[fleche] (12)--(11) {};
\draw[fleche] (13)--(1) {};
\draw[fleche] (13)--(1) {};
\draw[fleche] (13) to [bend left=45](2) {};
\draw[fleche] (13)--(4) {};
\draw[fleche] (13)--(7) {};
\draw[fleche] (13)--(10) {};
\draw[fleche] (13)--(12) {};
\draw[fleche] (14)--(1) {};
\draw[fleche] (14)--(13) {};
\end{tikzpicture}
\hspace{1.4cm}
\begin{tikzpicture}[xscale=\xscale,yscale=\yscale]
\tikzstyle{fleche}=[-,>=latex]
\tikzstyle{noeud}=[draw,circle,fill=black,scale=\nodescale*1]
\node[noeud] (0) at ({-0.5},{0}) {};
\tikzstyle{noeud}=[draw,circle,fill=black,scale=\nodescale*1]
\node[noeud] (1) at ({-0.5},{1}) {};
\tikzstyle{noeud}=[draw,circle,fill=black,scale=\nodescale*1]
\node[noeud] (2) at ({-0.5},{2}) {};
\tikzstyle{noeud}=[draw,circle,fill=black,scale=\nodescale*1]
\node[noeud] (3) at ({-0.5},{3}) {};
\tikzstyle{noeud}=[draw,circle,fill=black,scale=\nodescale*1]
\node[noeud] (4) at ({-0.5},{4}) {};
\tikzstyle{noeud}=[draw,circle,fill=black,scale=\nodescale*1]
\node[noeud] (5) at ({-0.5},{5}) {};
\tikzstyle{noeud}=[draw,circle,fill=black,scale=\nodescale*1]
\node[noeud] (6) at ({-0.5},{6}) {};
\tikzstyle{noeud}=[draw,circle,fill=black,scale=\nodescale*1]
\node[noeud] (7) at ({-0.5},{7}) {};
\draw[fleche] (1)--(0) {};
\draw[fleche] (2) to [bend left=45](0) {};
\draw[fleche] (2)--(1) {};
\draw[fleche] (3) to [bend left=45](0) {};
\draw[fleche] (3) to [bend right=45](1) {};
\draw[fleche] (3)--(2) {};
\draw[fleche] (4)--(3) {};
\draw[fleche] (5)--(4) {};
\draw[fleche] (6) to [bend left=45](1) {};
\draw[fleche] (6) to [bend right=45](1) {};
\draw[fleche] (6) to [bend left=45](1) {};
\draw[fleche] (6) to [bend right=45](2) {};
\draw[fleche] (6) to [bend left=45](3) {};
\draw[fleche] (6) to [bend right=45](4) {};
\draw[fleche] (6)--(5) {};
\draw[fleche] (7) to [bend left=45](1) {};
\draw[fleche] (7)--(6) {};
\end{tikzpicture}
\caption{A random tree of size $50$ (top left), its self-nested approximation tree of size $41$ (top right), the nonlinear DAG reduction of the initial tree (bottom left), and the linear DAG reduction of the self-nested approximation (bottom right).}
\label{fig:exav}
\end{figure}

\begin{proposition}
\label{prop:complav}
$\widehat{D}$ can be computed in $O(\#E_\equiv)$-time from the DAG reduction $D=(V_\equiv,E_\equiv,N)$ of $\tau$.
\end{proposition}

In \cite{GF2010}, the authors propose to approximate a tree by its Nearest Embedding Self-nested Tree (NEST), i.e., the self-nested tree obtained from the initial data by adding a minimal number of vertices. The NEST of the tree of Figure~\ref{fig:exav} (top left) is displayed in Fi\-gu\-re~\ref{fig:exnest} for illustration purposes. They establish in \cite[Theorem 1]{GF2010} that the NEST can be computed in $O(\height(\tau)^2 \deg(\tau))$. We prove from numerical simulations that our averaging algorithm achieves better approximation errors (on average approximately $30$ times lower for a tree of size $400$, see Figure~\ref{fig:compnesterr}) while it requires much less computation time (on average approximately $50$ times lower for a tree of size $800$, see Figure~\ref{fig:compnesttime}).

\begin{figure}[h]
\centering
\def\xscale{0.12}
\def\yscale{0.5}
\def\nodescale{0.2}
\begin{tikzpicture}[xscale=\xscale,yscale=\yscale]
\tikzstyle{fleche}=[-,>=latex]
\tikzstyle{noeudblack}=[draw,circle,fill=black,scale=\nodescale]
\def\localnodescalea{1}
\node[noeudblack,scale=\localnodescalea] (4545774200) at ({0},{0}) {};
\node[noeudblack,scale=\localnodescalea] (4545774704) at ({-30.5},{-1.0}) {};
\node[noeudblack,scale=\localnodescalea] (4545774984) at ({-31.5},{-1.8408964152537146}) {};
\node[noeudblack,scale=\localnodescalea] (4545773752) at ({-30.5},{-1.8408964152537146}) {};
\node[noeudblack,scale=\localnodescalea] (4545774928) at ({-29.5},{-1.8408964152537146}) {};
\node[noeudblack,scale=\localnodescalea] (4545774760) at ({1.5},{-1.0}) {};
\node[noeudblack,scale=\localnodescalea] (4545775264) at ({-27.5},{-1.8408964152537146}) {};
\node[noeudblack,scale=\localnodescalea] (4545774648) at ({-28.5},{-2.600732100905307}) {};
\node[noeudblack,scale=\localnodescalea] (4545775096) at ({-27.5},{-2.600732100905307}) {};
\node[noeudblack,scale=\localnodescalea] (4545775208) at ({-26.5},{-2.600732100905307}) {};
\node[noeudblack,scale=\localnodescalea] (4545774424) at ({-24.5},{-1.8408964152537146}) {};
\node[noeudblack,scale=\localnodescalea] (4545775320) at ({-25.5},{-2.600732100905307}) {};
\node[noeudblack,scale=\localnodescalea] (4545775376) at ({-24.5},{-2.600732100905307}) {};
\node[noeudblack,scale=\localnodescalea] (4545775432) at ({-23.5},{-2.600732100905307}) {};
\node[noeudblack,scale=\localnodescalea] (4545775152) at ({-21.5},{-1.8408964152537146}) {};
\node[noeudblack,scale=\localnodescalea] (4545775544) at ({-22.5},{-2.600732100905307}) {};
\node[noeudblack,scale=\localnodescalea] (4545775600) at ({-21.5},{-2.600732100905307}) {};
\node[noeudblack,scale=\localnodescalea] (4545775656) at ({-20.5},{-2.600732100905307}) {};
\node[noeudblack,scale=\localnodescalea] (4545775488) at ({-17.5},{-1.8408964152537146}) {};
\node[noeudblack,scale=\localnodescalea] (4545775768) at ({-19.5},{-2.600732100905307}) {};
\node[noeudblack,scale=\localnodescalea] (4545775824) at ({-18.5},{-2.600732100905307}) {};
\node[noeudblack,scale=\localnodescalea] (4545775880) at ({-16.5},{-2.600732100905307}) {};
\node[noeudblack,scale=\localnodescalea] (4545775992) at ({-17.5},{-3.3078388820918545}) {};
\node[noeudblack,scale=\localnodescalea] (4545776048) at ({-16.5},{-3.3078388820918545}) {};
\node[noeudblack,scale=\localnodescalea] (4545776104) at ({-15.5},{-3.3078388820918545}) {};
\node[noeudblack,scale=\localnodescalea] (4545775712) at ({-8.5},{-1.8408964152537146}) {};
\node[noeudblack,scale=\localnodescalea] (4545776160) at ({-14.5},{-2.600732100905307}) {};
\node[noeudblack,scale=\localnodescalea] (4545776216) at ({-13.5},{-2.600732100905307}) {};
\node[noeudblack,scale=\localnodescalea] (4545776272) at ({-11.5},{-2.600732100905307}) {};
\node[noeudblack,scale=\localnodescalea] (4545776384) at ({-12.5},{-3.3078388820918545}) {};
\node[noeudblack,scale=\localnodescalea] (4545776440) at ({-11.5},{-3.3078388820918545}) {};
\node[noeudblack,scale=\localnodescalea] (4545776496) at ({-10.5},{-3.3078388820918545}) {};
\node[noeudblack,scale=\localnodescalea] (4545776328) at ({-8.5},{-2.600732100905307}) {};
\node[noeudblack,scale=\localnodescalea] (4545776608) at ({-9.5},{-3.3078388820918545}) {};
\node[noeudblack,scale=\localnodescalea] (4545776664) at ({-8.5},{-3.3078388820918545}) {};
\node[noeudblack,scale=\localnodescalea] (4545776720) at ({-7.5},{-3.3078388820918545}) {};
\node[noeudblack,scale=\localnodescalea] (4545776552) at ({-4.5},{-2.600732100905307}) {};
\node[noeudblack,scale=\localnodescalea] (4545776832) at ({-6.5},{-3.3078388820918545}) {};
\node[noeudblack,scale=\localnodescalea] (4545776888) at ({-5.5},{-3.3078388820918545}) {};
\node[noeudblack,scale=\localnodescalea] (4545776944) at ({-3.5},{-3.3078388820918545}) {};
\node[noeudblack,scale=\localnodescalea] (4545777056) at ({-4.5},{-3.9765791870682765}) {};
\node[noeudblack,scale=\localnodescalea] (4545777112) at ({-3.5},{-3.9765791870682765}) {};
\node[noeudblack,scale=\localnodescalea] (4545777168) at ({-2.5},{-3.9765791870682765}) {};
\node[noeudblack,scale=\localnodescalea] (4545775936) at ({6.5},{-1.8408964152537146}) {};
\node[noeudblack,scale=\localnodescalea] (4545777000) at ({-1.5},{-2.600732100905307}) {};
\node[noeudblack,scale=\localnodescalea] (4545777224) at ({0.5},{-2.600732100905307}) {};
\node[noeudblack,scale=\localnodescalea] (4545777336) at ({-0.5},{-3.3078388820918545}) {};
\node[noeudblack,scale=\localnodescalea] (4545777392) at ({0.5},{-3.3078388820918545}) {};
\node[noeudblack,scale=\localnodescalea] (4545777448) at ({1.5},{-3.3078388820918545}) {};
\node[noeudblack,scale=\localnodescalea] (4545777280) at ({8.5},{-2.600732100905307}) {};
\node[noeudblack,scale=\localnodescalea] (4545777560) at ({2.5},{-3.3078388820918545}) {};
\node[noeudblack,scale=\localnodescalea] (4545777616) at ({3.5},{-3.3078388820918545}) {};
\node[noeudblack,scale=\localnodescalea] (4545818696) at ({5.5},{-3.3078388820918545}) {};
\node[noeudblack,scale=\localnodescalea] (4545818808) at ({4.5},{-3.9765791870682765}) {};
\node[noeudblack,scale=\localnodescalea] (4545818864) at ({5.5},{-3.9765791870682765}) {};
\node[noeudblack,scale=\localnodescalea] (4545818920) at ({6.5},{-3.9765791870682765}) {};
\node[noeudblack,scale=\localnodescalea] (4545818752) at ({8.5},{-3.3078388820918545}) {};
\node[noeudblack,scale=\localnodescalea] (4545819032) at ({7.5},{-3.9765791870682765}) {};
\node[noeudblack,scale=\localnodescalea] (4545819088) at ({8.5},{-3.9765791870682765}) {};
\node[noeudblack,scale=\localnodescalea] (4545819144) at ({9.5},{-3.9765791870682765}) {};
\node[noeudblack,scale=\localnodescalea] (4545818976) at ({12.5},{-3.3078388820918545}) {};
\node[noeudblack,scale=\localnodescalea] (4545819256) at ({10.5},{-3.9765791870682765}) {};
\node[noeudblack,scale=\localnodescalea] (4545819312) at ({11.5},{-3.9765791870682765}) {};
\node[noeudblack,scale=\localnodescalea] (4545819368) at ({13.5},{-3.9765791870682765}) {};
\node[noeudblack,scale=\localnodescalea] (4545819480) at ({12.5},{-4.615522291314549}) {};
\node[noeudblack,scale=\localnodescalea] (4545819536) at ({13.5},{-4.615522291314549}) {};
\node[noeudblack,scale=\localnodescalea] (4545819592) at ({14.5},{-4.615522291314549}) {};
\node[noeudblack,scale=\localnodescalea] (4545776776) at ({23.5},{-1.8408964152537146}) {};
\node[noeudblack,scale=\localnodescalea] (4545819200) at ({23.5},{-2.600732100905307}) {};
\node[noeudblack,scale=\localnodescalea] (4545819648) at ({15.5},{-3.3078388820918545}) {};
\node[noeudblack,scale=\localnodescalea] (4545819704) at ({17.5},{-3.3078388820918545}) {};
\node[noeudblack,scale=\localnodescalea] (4545819816) at ({16.5},{-3.9765791870682765}) {};
\node[noeudblack,scale=\localnodescalea] (4545819872) at ({17.5},{-3.9765791870682765}) {};
\node[noeudblack,scale=\localnodescalea] (4545819928) at ({18.5},{-3.9765791870682765}) {};
\node[noeudblack,scale=\localnodescalea] (4545819760) at ({25.5},{-3.3078388820918545}) {};
\node[noeudblack,scale=\localnodescalea] (4545820040) at ({19.5},{-3.9765791870682765}) {};
\node[noeudblack,scale=\localnodescalea] (4545820096) at ({20.5},{-3.9765791870682765}) {};
\node[noeudblack,scale=\localnodescalea] (4545820152) at ({22.5},{-3.9765791870682765}) {};
\node[noeudblack,scale=\localnodescalea] (4545820264) at ({21.5},{-4.615522291314549}) {};
\node[noeudblack,scale=\localnodescalea] (4545820320) at ({22.5},{-4.615522291314549}) {};
\node[noeudblack,scale=\localnodescalea] (4545820376) at ({23.5},{-4.615522291314549}) {};
\node[noeudblack,scale=\localnodescalea] (4545820208) at ({25.5},{-3.9765791870682765}) {};
\node[noeudblack,scale=\localnodescalea] (4545820488) at ({24.5},{-4.615522291314549}) {};
\node[noeudblack,scale=\localnodescalea] (4545820544) at ({25.5},{-4.615522291314549}) {};
\node[noeudblack,scale=\localnodescalea] (4545820600) at ({26.5},{-4.615522291314549}) {};
\node[noeudblack,scale=\localnodescalea] (4545820432) at ({29.5},{-3.9765791870682765}) {};
\node[noeudblack,scale=\localnodescalea] (4545820712) at ({27.5},{-4.615522291314549}) {};
\node[noeudblack,scale=\localnodescalea] (4545820768) at ({28.5},{-4.615522291314549}) {};
\node[noeudblack,scale=\localnodescalea] (4545820824) at ({30.5},{-4.615522291314549}) {};
\node[noeudblack,scale=\localnodescalea] (4545820936) at ({29.5},{-5.230310444265813}) {};
\node[noeudblack,scale=\localnodescalea] (4545820992) at ({30.5},{-5.230310444265813}) {};
\node[noeudblack,scale=\localnodescalea] (4545821048) at ({31.5},{-5.230310444265813}) {};
\draw[fleche] (4545774200)--(4545774704) {};
\draw[fleche] (4545774200)--(4545774760) {};
\draw[fleche] (4545774704)--(4545774984) {};
\draw[fleche] (4545774704)--(4545773752) {};
\draw[fleche] (4545774704)--(4545774928) {};
\draw[fleche] (4545774760)--(4545775264) {};
\draw[fleche] (4545774760)--(4545774424) {};
\draw[fleche] (4545774760)--(4545775152) {};
\draw[fleche] (4545774760)--(4545775488) {};
\draw[fleche] (4545774760)--(4545775712) {};
\draw[fleche] (4545774760)--(4545775936) {};
\draw[fleche] (4545774760)--(4545776776) {};
\draw[fleche] (4545775264)--(4545774648) {};
\draw[fleche] (4545775264)--(4545775096) {};
\draw[fleche] (4545775264)--(4545775208) {};
\draw[fleche] (4545774424)--(4545775320) {};
\draw[fleche] (4545774424)--(4545775376) {};
\draw[fleche] (4545774424)--(4545775432) {};
\draw[fleche] (4545775152)--(4545775544) {};
\draw[fleche] (4545775152)--(4545775600) {};
\draw[fleche] (4545775152)--(4545775656) {};
\draw[fleche] (4545775488)--(4545775768) {};
\draw[fleche] (4545775488)--(4545775824) {};
\draw[fleche] (4545775488)--(4545775880) {};
\draw[fleche] (4545775880)--(4545775992) {};
\draw[fleche] (4545775880)--(4545776048) {};
\draw[fleche] (4545775880)--(4545776104) {};
\draw[fleche] (4545775712)--(4545776160) {};
\draw[fleche] (4545775712)--(4545776216) {};
\draw[fleche] (4545775712)--(4545776272) {};
\draw[fleche] (4545775712)--(4545776328) {};
\draw[fleche] (4545775712)--(4545776552) {};
\draw[fleche] (4545776272)--(4545776384) {};
\draw[fleche] (4545776272)--(4545776440) {};
\draw[fleche] (4545776272)--(4545776496) {};
\draw[fleche] (4545776328)--(4545776608) {};
\draw[fleche] (4545776328)--(4545776664) {};
\draw[fleche] (4545776328)--(4545776720) {};
\draw[fleche] (4545776552)--(4545776832) {};
\draw[fleche] (4545776552)--(4545776888) {};
\draw[fleche] (4545776552)--(4545776944) {};
\draw[fleche] (4545776944)--(4545777056) {};
\draw[fleche] (4545776944)--(4545777112) {};
\draw[fleche] (4545776944)--(4545777168) {};
\draw[fleche] (4545775936)--(4545777000) {};
\draw[fleche] (4545775936)--(4545777224) {};
\draw[fleche] (4545775936)--(4545777280) {};
\draw[fleche] (4545777224)--(4545777336) {};
\draw[fleche] (4545777224)--(4545777392) {};
\draw[fleche] (4545777224)--(4545777448) {};
\draw[fleche] (4545777280)--(4545777560) {};
\draw[fleche] (4545777280)--(4545777616) {};
\draw[fleche] (4545777280)--(4545818696) {};
\draw[fleche] (4545777280)--(4545818752) {};
\draw[fleche] (4545777280)--(4545818976) {};
\draw[fleche] (4545818696)--(4545818808) {};
\draw[fleche] (4545818696)--(4545818864) {};
\draw[fleche] (4545818696)--(4545818920) {};
\draw[fleche] (4545818752)--(4545819032) {};
\draw[fleche] (4545818752)--(4545819088) {};
\draw[fleche] (4545818752)--(4545819144) {};
\draw[fleche] (4545818976)--(4545819256) {};
\draw[fleche] (4545818976)--(4545819312) {};
\draw[fleche] (4545818976)--(4545819368) {};
\draw[fleche] (4545819368)--(4545819480) {};
\draw[fleche] (4545819368)--(4545819536) {};
\draw[fleche] (4545819368)--(4545819592) {};
\draw[fleche] (4545776776)--(4545819200) {};
\draw[fleche] (4545819200)--(4545819648) {};
\draw[fleche] (4545819200)--(4545819704) {};
\draw[fleche] (4545819200)--(4545819760) {};
\draw[fleche] (4545819704)--(4545819816) {};
\draw[fleche] (4545819704)--(4545819872) {};
\draw[fleche] (4545819704)--(4545819928) {};
\draw[fleche] (4545819760)--(4545820040) {};
\draw[fleche] (4545819760)--(4545820096) {};
\draw[fleche] (4545819760)--(4545820152) {};
\draw[fleche] (4545819760)--(4545820208) {};
\draw[fleche] (4545819760)--(4545820432) {};
\draw[fleche] (4545820152)--(4545820264) {};
\draw[fleche] (4545820152)--(4545820320) {};
\draw[fleche] (4545820152)--(4545820376) {};
\draw[fleche] (4545820208)--(4545820488) {};
\draw[fleche] (4545820208)--(4545820544) {};
\draw[fleche] (4545820208)--(4545820600) {};
\draw[fleche] (4545820432)--(4545820712) {};
\draw[fleche] (4545820432)--(4545820768) {};
\draw[fleche] (4545820432)--(4545820824) {};
\draw[fleche] (4545820824)--(4545820936) {};
\draw[fleche] (4545820824)--(4545820992) {};
\draw[fleche] (4545820824)--(4545821048) {};
\end{tikzpicture}
\caption{The NEST of the tree of Figure~\ref{fig:exav} (top left). It has $92$ vertices, which corresponds to an increase of $84\%$ of the size of the initial tree.}
\label{fig:exnest}
\end{figure}
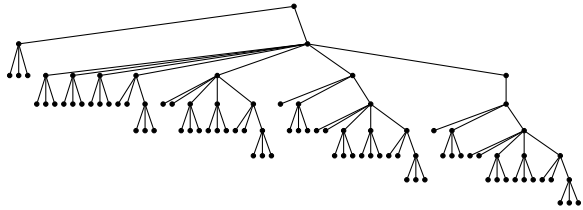

\begin{figure}[h]
\centering
\includegraphics[width=8.5cm]{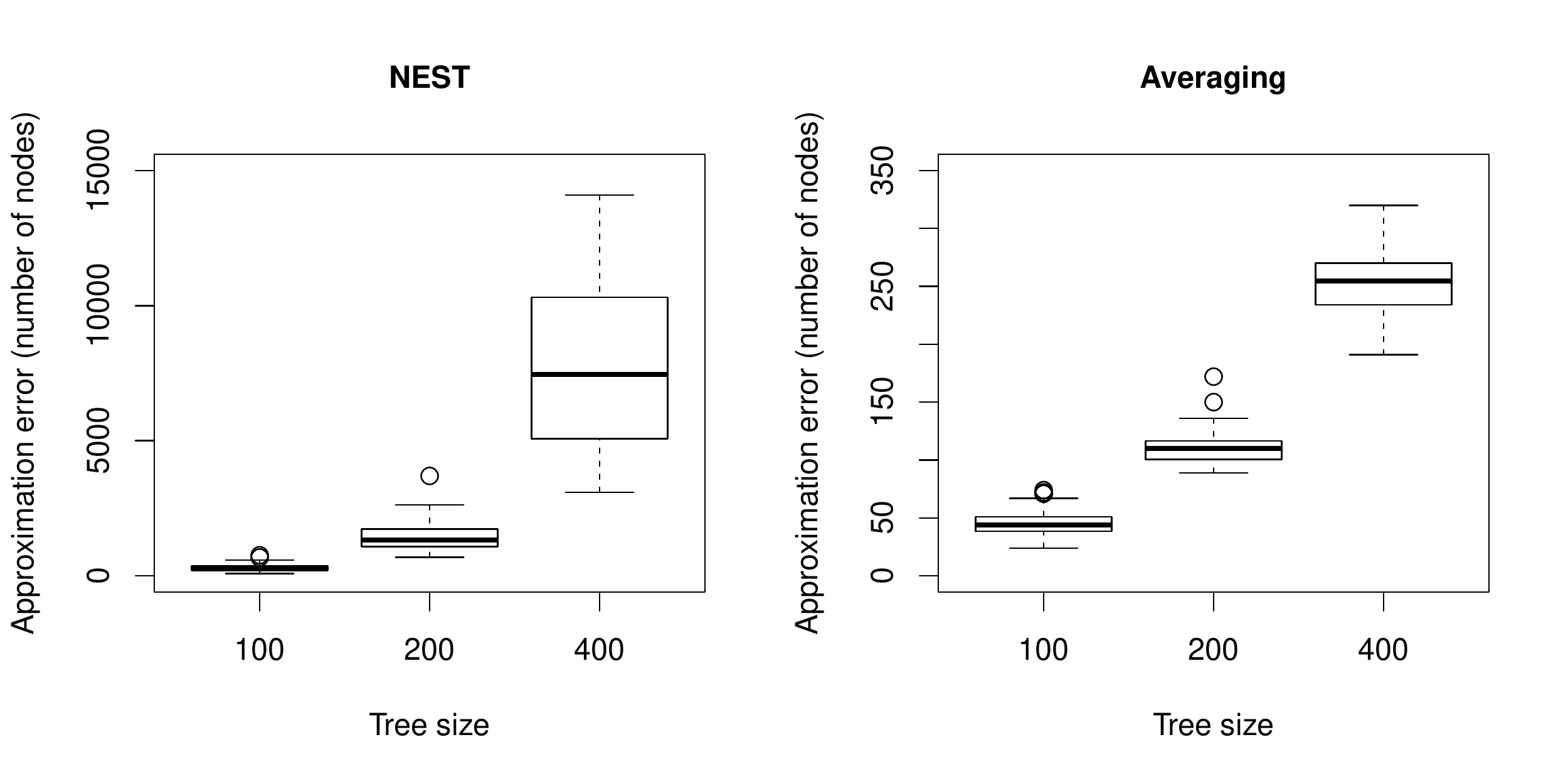}
\caption{Comparison between NEST approximation algorithm (left) and the averaging method proposed in this paper (right) in terms of approximation error.}
\label{fig:compnesterr}
\end{figure}

\begin{figure}[h]
\centering
\includegraphics[width=8.5cm]{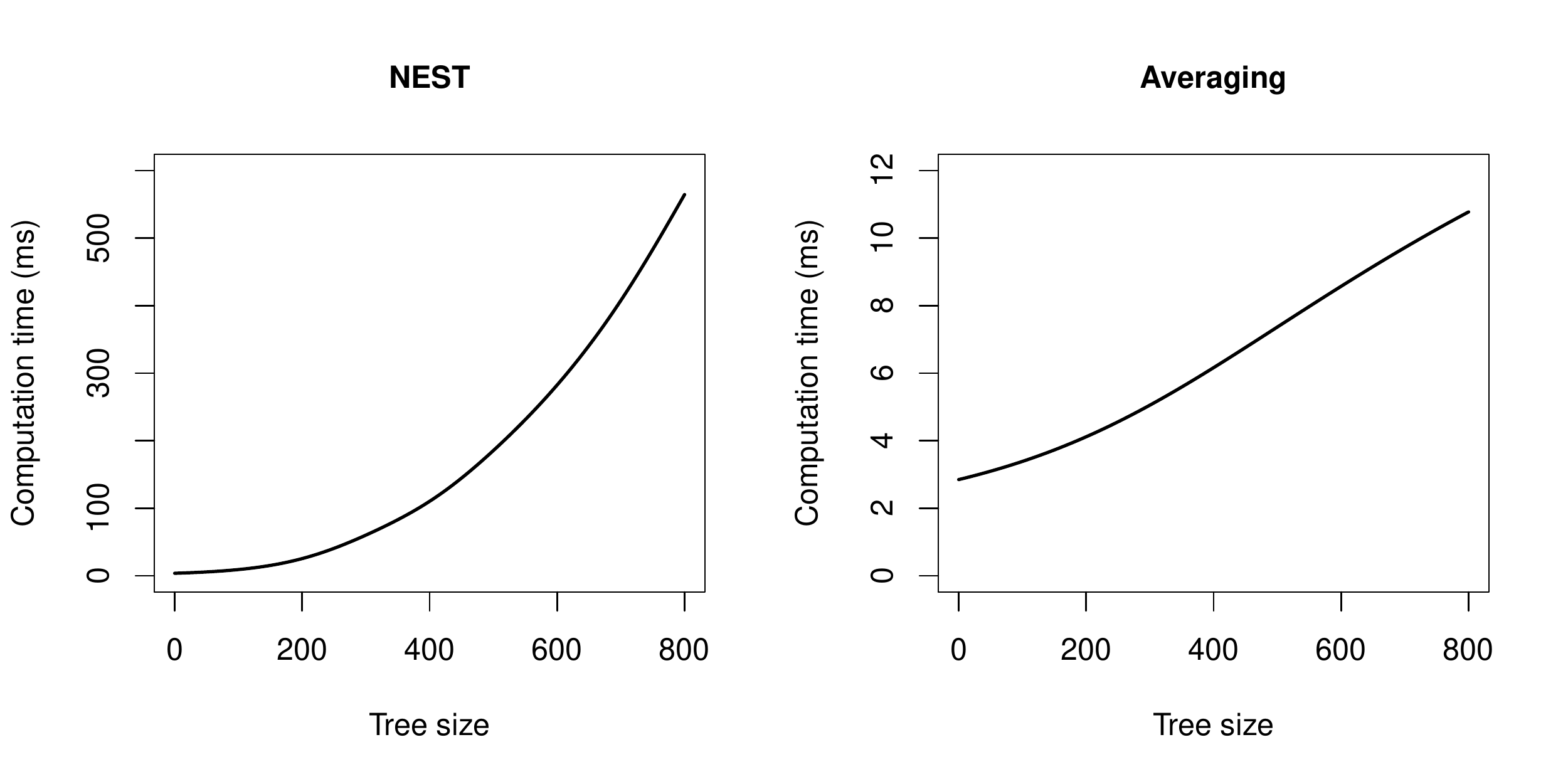}
\caption{Comparison between NEST approximation algorithm (left) and the averaging method proposed in this paper (right) in terms of computation time.}
\label{fig:compnesttime}
\end{figure}

\subsection{Error upper bound in approximation}

We establish the optimal bound of the approximation error of a tree by a self-nested one in terms of edit distance $\delta$ in the following result.

\begin{proposition}
\label{prop:wc}
For any $H\geq2$ and $d$ large enough (greater than a constant depending on $H$),
$$ \max_{t\in\mathbb{T}_{\leq H,\leq d}} \min_{\tau\in\mathbb{T}^{sn}} \delta(t,\tau) = \left\lfloor\frac{d}{2}\right\rfloor \times \left\lceil \frac{d}{2}\right\rceil \times d^{H-2} .$$
In addition, this worst case is reached for the nonlinear DAG of Figure~\ref{fig:wc} (left).
\end{proposition}

\begin{figure}
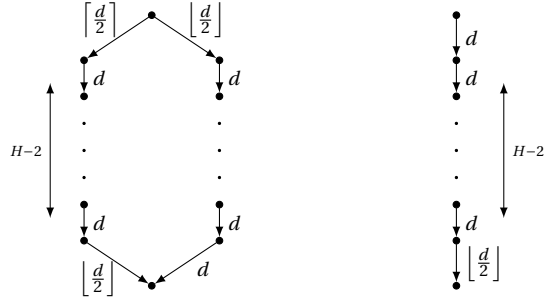

\centering
\dagTcal
\caption{Nonlinear DAG of height $H$ and outdegree $d$ (left) for which the best linear DAG approximation (right) achieves the worst error in terms of edit distance $\delta$.}
\label{fig:wc}
\end{figure}

The diameter of the state space $\mathbb{T}_{\leq H,\leq d}$ is of order $d^H$ (indeed, the largest tree of this family is the full $d$-tree, while the smallest tree is reduced to a unique root vertex). As a consequence and in light of Proposition~\ref{prop:wc}, the largest area without any self-nested tree is a ball with relative radius
\begin{eqnarray*}
\frac{\left\lfloor\frac{d}{2}\right\rfloor \times \left\lceil \frac{d}{2}\right\rceil \times d^{H-2}}{d^H} &=&\frac{\left\lfloor\frac{d}{2}\right\rfloor \times \left\lceil \frac{d}{2}\right\rceil}{d^2} \\
&=& \frac{1}{4}+\frac{1}{4 d^2}\mathbb{1}_{2\mathbb{N}+1}(d) ~ \simeq ~ \frac{1}{4}.
\end{eqnarray*}
This result is especially noteworthy considering the very low frequency of self-nested trees compared to unordered trees (see Table~\ref{tableau1} and Proposition~\ref{prop2}). Nevertheless, we would like to emphasize that this negative result (exponential upper bound of the error) is far from representing the average behavior of the averaging algorithm presented in Figure~\ref{fig:compnesterr} (right).


\section{Fast prediction of edit distance}
\label{s:fastpred}

This section is dedicated to an example of application of the averaging algorithm: we show that it can be used for fast prediction of edit distance between two trees. Given two trees $\tau_1$ and $\tau_2$, we propose to estimate $\delta(\tau_1,\tau_2)$ by $\delta(\widehat{\tau}_1,\widehat{\tau}_2)$, where $\widehat{\tau}_i$ is the self-nested approximation of $\tau_i$. Including the approximation step, computing the edit distance between $\widehat{\tau}_1$ and $\widehat{\tau}_2$ is on average $10$ times faster than computing $\delta(\tau_1,\tau_2)$ from trees with $1\,000$ vertices (see Figure~\ref{fig:comptime_approx} for more details on computation times), which represents a significant gain even for trees of reasonable size.

\begin{figure}[h]
\centering
\includegraphics[width=8.5cm]{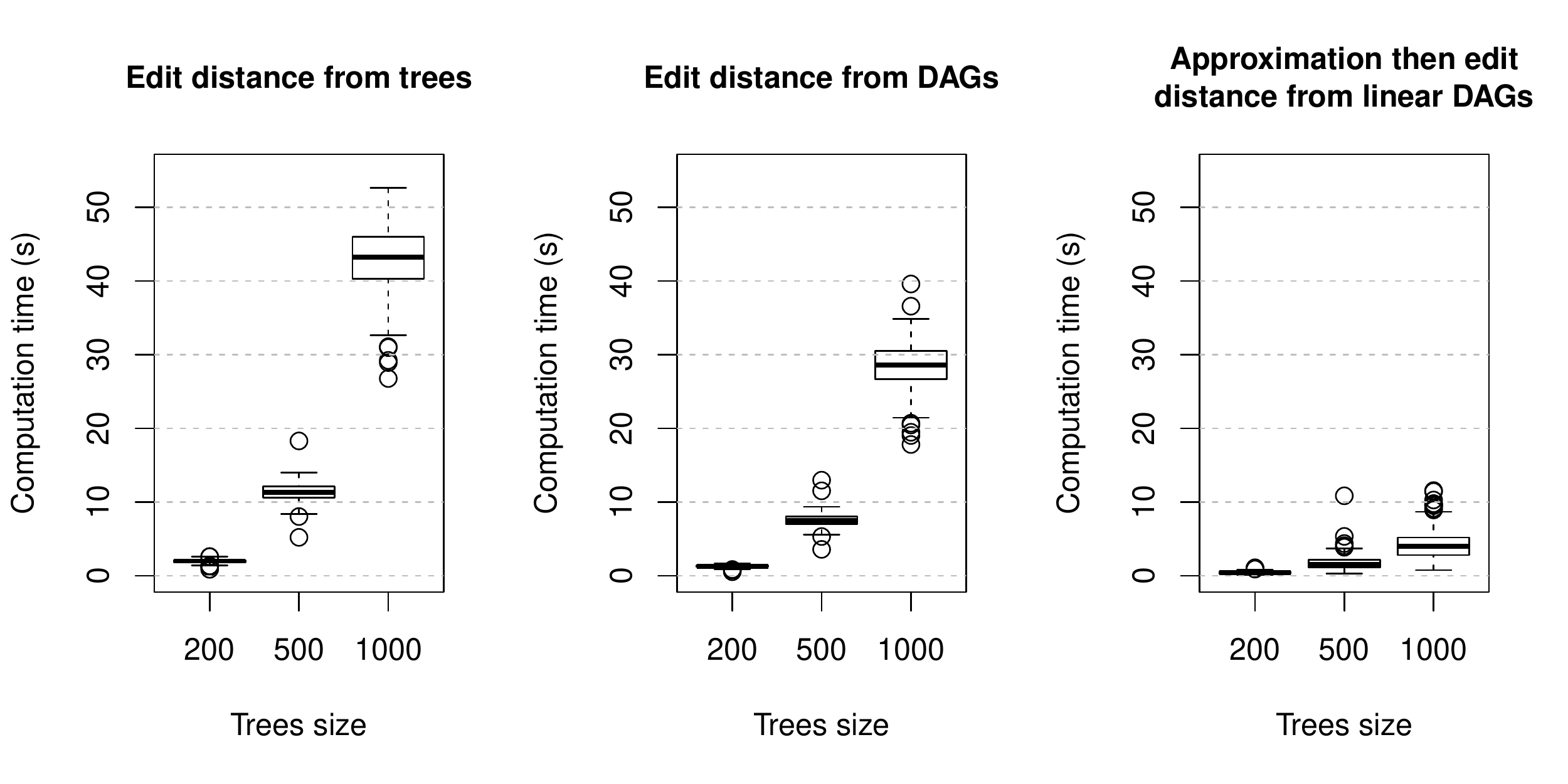}
\caption{Estimation of the computation time of the edit distance between two trees from the trees (left), from their DAG reduction (middle), and from their linear DAG approximation (including computation of the approximations, right).}
\label{fig:comptime_approx}
\end{figure}

Nevertheless, $\delta(\widehat{\tau}_1,\widehat{\tau}_2)$ does not always provide a good estimate of $\delta(\tau_1,\tau_2)$: on average, the relative error is of order $-50\%$ (see Figure~\ref{fig:learning} (left)), i.e., $\delta(\widehat{\tau}_1,\widehat{\tau}_2)$ strongly underestimates $\delta(\tau_1,\tau_2)$. To improve this relative error, we correct the predictor via a learning algorithm. First, we generate from random trees $(t_1,t_2)$ a training dataset containing $2\,500$ replicates of $\delta(t_1,t_2)$ and $\delta(\widehat{t}_1,\widehat{t}_2)$. Then we learn a prediction rule of $\delta(t_1,t_2)$ from $\delta(\widehat{t}_1,\widehat{t}_2)$ with a linear model as implemented in the \verb+lm+ function in \verb+R+ \cite{chambers:1992:chap4}. Given two trees $\tau_1$ and $\tau_2$, we first compute $\delta(\widehat{\tau}_1,\widehat{\tau}_2)$ and then we predict $\delta(\tau_1,\tau_2)$ with the aforementioned prediction rule. The results on a test dataset of size $1\,500$ are presented in Figure~\ref{fig:learning} (middle). The error is null on average but presents a large variance making the prediction not reliable.

This can be corrected by adding explanatory variables to the learning step. In the learning dataset, we add the following features: size, height, outdegree and Strahler number of $t_i$ and $\widehat{t}_i$. It should be noted that these quantities can be computed without adding any computation time to the whole procedure during the computation of the DAG reductions. Thus, this does not affect the speed of our prediction algorithm. The results are presented in Figure~\ref{fig:learning} (right). The error is null on average with a small variance: in $50\%$ of cases, the prediction error is between $-6.5\%$ and $6.9\%$. The most significant variables ($p$-values less than $2.10^{-16}$) are $\delta(\widehat{t}_1,\widehat{t}_2)$ and the sizes of the 4 trees. In addition, the prediction rule learned from the same training dataset without $\delta(\widehat{t}_1,\widehat{t}_2)$ and the additional informations on the self-nested approximations (size, height, outdegree and Strahler number) achieve worse results in $66.9\%$ of cases. This shows that self-nested approximations can be used to obtain a fast and accurate prediction of the edit distance between two trees.

\begin{figure}[h]
\centering
\includegraphics[width=8.5cm]{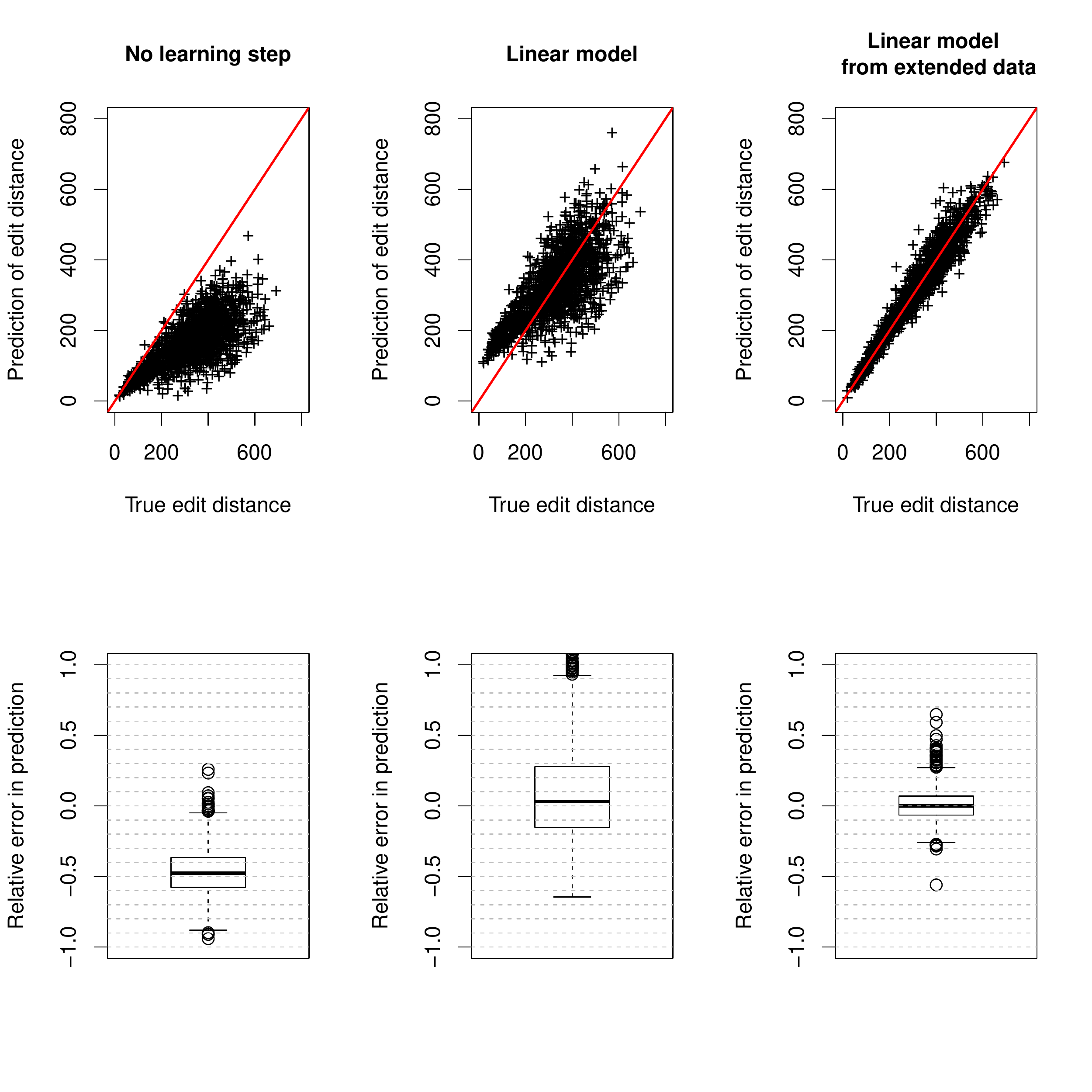}
\caption{Prediction (top) and relative error in prediction (bottom) of the edit distance between two trees from their li\-ne\-ar DAG approximations (left), from a linear model learned with edit distances between li\-near DAG approximations (middle), and from a linear model learned with edit distances between linear DAG approximations and additional informations on the trees (right).}
\label{fig:learning}
\end{figure}

\paragraph{Acknowledgments}  The authors thank the reviewers for their careful reading of the manuscript and their insightful comments. R.A. and C.G. were supported by the European Union's Horizon 2020 project ROMI.


\newpage 
\bibliographystyle{acm}
\bibliography{main}

\newpage
\appendix

\begin{center}\large Appendix to the paper\\
\Large Approximation of trees by self-nested trees
\end{center}


\begin{center}
Romain Aza\"{\i}s,
Jean-Baptiste Durand, and Christophe Godin
\end{center}

This supplementary file contains the proofs of all the original theoretical results presented in the main paper.


\section{Proof of Proposition~\ref{prop:bu:compl}}

Computing a bottom-up function requires to traverse the tree $\tau$ from the leaves to the root in $O(\#\tau)$. From the DAG, we have, with the notation of Subsection \ref{ss:treereduction},
\begin{equation}
f\big( \tree(D[ (h_1, i )] ) \big) = \Phi\Bigg(\,\Big(\underbrace{f(\tree(D[(h_2,j)] ))}_{N((h_1,i),(h_2,j))\,\text{times}}\Big)_{(h_2,j)}\, \Bigg) .
\label{eq:recur:dag}
\end{equation}
It should be noted that \eqref{eq:recur:dag} holds only if $\Phi$ is invariant under permutation of arguments. In addition,
$$ \sum_{h_2=0}^{h_1-1}\sum_{j=1}^{M_{h_2}} N((h_1,i),(h_2,j)) \leq \deg(\tau).$$
As a consequence, computing $f(\tau)$ from its DAG reduction requires to traverse the DAG with at most $\deg(\tau)$ operations on each vertex, which states the result.


\section{Proof of Proposition \ref{prop:delta:complexity}}

\subsection{Preliminaries}

\begin{lemma}\label{prop:delta:dist}
$\delta$ defines a distance function on the space of unordered rooted trees $\mathbb{T}$.
\end{lemma}
\begin{proof}
The separation axiom is obviously satisfied by $\delta$ because of its definition as a cardinality. In addition, $\tau_1\equiv\tau_2$ if and only if the empty script $\emptyset$ satisfies $\tau_1^\emptyset\equiv\tau_2$, so that the coincidence axiom is checked. Symmetry is obvious by applying in the reverse order the reverse operations of a script $s$. Finally, if $s$ ($\sigma$, respectively) denotes a minimum-length script to transform $\tau_1$ into $\tau_2$ ($\tau_2$ into $\tau_3$, respectively), the script $s\sigma$ obtained as the concatenation of both these scripts transforms $\tau_1$ into $\tau_3$. The triangle inequality is thus satisfied,
$$\delta(\tau_1,\tau_3)\leq\#(s\sigma)=\#s+\#\sigma=\delta(\tau_1,\tau_2)+\delta(\tau_2,\tau_3),$$
which yields the expected result.
\end{proof}

Here we address the issue of equivalence between edit distance and tree mapping cost using the particular edit distance $\delta$. Such equivalence has been discussed in \cite{doi:10.1142/S0218001488000157,Zhang:1996:CED} in the context of other edit distances.

\paragraph{Mapping} Let $\tau_1$ and $\tau_2$ be two trees. Suppose that we have a numbering of the vertices for each tree. Since we are concerned with unordered trees, we can fix an arbitrary order for each of the vertex in the tree and then use left-to-right postorder numbering or left-to-right preorder numbering. A mapping $\mathcal{M}$ from $\tau_1$ to $\tau_2$ is a set of couples $i\to j$, $1\leq i\leq\#\tau_1$ and $1\leq j\leq\#\tau_2$, satisfying (see \cite[2.3.2 Editing Distance Mappings]{Zhang:1996:CED}), for any $i_1\to j_1$ and $i_2\to j_2$ in $\mathcal{M}$, the following assumptions:
\begin{itemize}
\item $i_1=i_2$ if and only if $j_1=j_2$;
\item vertex $i_1$ in $\tau_1$ is an ancestor of vertex $i_2$ in $\tau_1$ if and only if vertex $j_1$ in $\tau_2$ is an ancestor of vertex $j_2$ in $\tau_2$.
\end{itemize}

\paragraph{Constrained tree mapping} Let $\tau_1$ and $\tau_2$ be two trees, $s$ be a script such that $\tau_1^s\equiv\tau_2$ and $\varphi$ a tree isomorphism between $\tau_1^s$ and $\tau_2$. The graph $\tau_1\cap\tau_1^s$ defines a tree embedded in $\tau_1$ because script $s$ only added and deleted leaves. As a consequence, the function $\widehat{\varphi}$ defined as $\varphi$ restricted to $\tau_1\cap\tau_1^s$ provides a tree mapping from $\tau_1$ to $\tau_2$ with $i\to j$ if and only if $\widehat{\varphi}(i)=j$. Of course, this is a particular tree mapping since it has been obtained from very special conditions. The main additional condition is the following: for any $i_1\to j_1$ and $i_2\to j_2$,
\begin{itemize}
\item vertex $i_1$ is the parent of vertex $i_2$ in $\tau_1$ if and only if vertex $j_1$ is the parent of vertex $j_2$ in $\tau_2$.
\end{itemize}
It is easy to see that this assumption is actually the only required additional constraint to define the class of constrained tree mappings involved in the computation of our constrained edit distance $\delta$. The equivalence between constrained mappings and $\delta$ may be stated as follows:
$$
\delta(\tau_1,\tau_2) = \min_{\mathcal{M}} \Big[ \#\{i\,:\,\nexists\,j~\text{s.t.}~i\to j\in\mathcal{M}\} + \#\{j\,:\,\nexists\,i~\text{s.t.}~i\to j\in\mathcal{M}\} \Big],
$$
where the minimum is taken over the set of all the possible constrained tree mappings. The equivalence between tree mapping cost and edit distance is a classical property used in the computation of the edit distance.

The mappings involved in our constrained edit distance have other properties related to some previous works of the literature. We present additional definitions, namely, the least common ancestor of two vertices and the constrained mappings presented in \cite{Zhang:1996:CED}.

\paragraph{Least common ancestor} The Least Common Ancestor (LCA) of two vertices $v$ and $w$ in a same tree is the lowest (i.e., least height) vertex that has both $v$ and $w$ as descendants. In other words, the LCA is the shared ancestor that is located farthest from the root. It should be noted that if $v$ is a descendant of $w$, $w$ is the LCA. 

\paragraph{Constrained mapping in Zhang's distance} Tanaka and Tanaka proposed in \cite{doi:10.1142/S0218001488000157} the following condition for mapping ordered labeled trees: disjoint subtrees should be mapped to disjoint subtrees. They showed that in some applications (e.g., classification tree comparison) this kind of mapping is more meaningful than more general edit distance mappings. Zhang investigated in \cite{Zhang:1996:CED} the problem of computing the edit distance associated with this kind of constrained mapping between unordered labeled trees. Precisely, a constrained mapping $\mathcal{M}$ between trees $\tau_1$ and $\tau_2$ in sense of Zhang is a mapping satisfying the additional condition \cite[3.1. Constrained Edit Distance Mappings]{Zhang:1996:CED}:
\begin{itemize}
\item Assume that $i_1\to j_1$, $i_2\to j_2$ and $i_3\to j_3$ are in $\mathcal{M}$. Let $v$ ($w$, respectively) be the LCA of vertices $i_1$ and $i_2$ in $\tau_1$ (of vertices $j_1$ and $j_2$ in $\tau_2$, respectively). $v$ is a proper ancestor of vertex $i_3$ in $\tau_1$ if and only if $w$ is a proper ancestor of vertex $j_3$ in $\tau_2$.
\end{itemize}

Let $\tau_1$ and $\tau_2$ be two trees and $\mathcal{M}$ a constrained tree mapping in sense of this paper. First, one may remark that the roots are necessarily mapped together. In addition, $\mathcal{M}$ satisfies all the conditions of constrained mappings imposed by Zhang in \cite{Zhang:1996:CED} and presented above.

\subsection{Reduction to the minimum cost flow problem}
\label{ss:mincostflow}

The edit distance between two trees $\tau_1$ and $\tau_2$ may be obtained from the recursive formula presented in Proposition~\ref{prop:calcul:ed} hereafter. In the sequel, $\mathcal{F}_t$ denotes the forest of all the subtrees of $t$ which root is a child of $\roottree(t)$, i.e., $\mathcal{F}_t$ is the list of the subtrees of $t$ that can be found just under its root. Furthermore, $S(n)$ denotes the set of permutations of $\{1,\dots,n\}$ and ${\binom{A}{n}}$ the set of subsets of $A$ with cardinality $n$.

\begin{proposition}\label{prop:calcul:ed}
Let $\tau_1$ and $\tau_2$ be two trees and $n=\min(\#\mathcal{F}_{\tau_1},\#\mathcal{F}_{\tau_2})$. The edit distance between $\tau_1$ and $\tau_2$ satisfies the following induction formula,
$$
\delta(\tau_1,\tau_2) = \min_{\{t_1^1,\dots,t_n^1\} \atop {\{t_1^2,\dots,t_n^2\} \atop \sigma}} \Bigg[\sum_{i=1}^n\delta(t_i^1,t_{\sigma(i)}^2)+\!\!\!\!\sum_{\theta\notin(t_1^1,\dots,t_n^1)}\!\!\!\delta(\theta,\emptyset) +\!\!\!\!\sum_{\theta\notin(t_1^2,\dots,t_n^2)}\!\!\!\delta(\emptyset,\theta)\Bigg],$$
where the minimum is taken over $\{t_1^1,\dots,t_n^1\}\in{\binom{\mathcal{F}_{T_1}}{n}}$, $\{t_1^2,\dots,t_n^2\}\in{\binom{\mathcal{F}_{T_2}}{n}}$ and $\sigma\in S(n)$, and the symbol $\emptyset$ stands for the empty tree. The formula is initialized with
$$\delta(\tau_1,\emptyset) = \# \tau_1\qquad\text{and}\qquad\delta(\emptyset,\tau_2)=\# \tau_2.$$
\end{proposition}
\begin{proof}
First, let us remark that a maximum number of subtrees of $\mathcal{F}_{\tau_1}$ should be mapped to subtrees of $\mathcal{F}_{\tau_2}$, because $\delta(\theta_1,\theta_2)<\delta(\theta_1,\emptyset)+\delta(\emptyset,\theta_2)$, for any trees $\theta_1$ and $\theta_2$. This maximum number is $n=\min(\#\mathcal{F}_{\tau_1},\#\mathcal{F}_{\tau_2})$. As a consequence, the minimal editing cost is obtained by considering all the possible mappings between $n$ subtrees of $\mathcal{F}_{\tau_1}$ and $n$ subtrees of $\mathcal{F}_{\tau_2}$. The subtrees that are not involved in a mapping are either deleted or added. We refer the reader to Figure~\ref{fig:calcul:ed}.
\end{proof}

\begin{figure}[h]
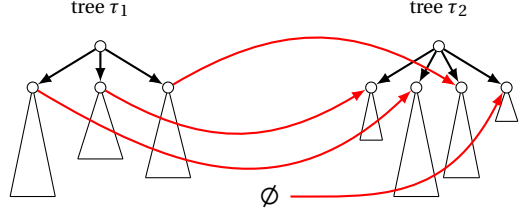

\centering\algoed
\caption{Schematic illustration of the recursive formula to compute the constrained edit distance $\delta$ between trees $\tau_1$ and $\tau_2$: the three subtrees of $\mathcal{F}_{\tau_1}$ are mapped to three subtrees of $\mathcal{F}_{\tau_2}$, while the empty tree $\emptyset$ is mapped to the fourth subtree of $\mathcal{F}_{\tau_2}$, i.e., all the vertices are added.
}
\label{fig:calcul:ed}
\end{figure}

In light of Proposition \ref{prop:calcul:ed} and Figure~\ref{fig:calcul:ed} and as in \cite[5. Algorithm and complexity]{Zhang:1996:CED}, each step in the recursive computation of the edit distance $\delta(\tau_1,\tau_2)$ between trees $\tau_1$ and $\tau_2$ reduces to the minimum cost maximum flow problem on a graph $G=(V,E)$ constructed as follows. First the set of vertices $V$ of $G$ is defined by
$$V=\left\{\text{source}\, ,\, \text{sink}\, ,\, \emptyset_{\tau_1}\, ,\, \emptyset_{\tau_2}\right\} \cup \mathcal{F}_{\tau_1} \cup \mathcal{F}_{\tau_2} .$$
The set $E$ of edges of $G$ is defined from:
\begin{itemize}
\item $\text{source}\to t_i^1$, $t_i^1\in\mathcal{F}_{\tau_1}$
\begin{itemize}
	\item capacity: $1$
	\item cost: $0$
\end{itemize}
\item $\text{source}\to \emptyset_{\tau_1}$
\begin{itemize}
	\item capacity: $\#\mathcal{F}_{\tau_1}-\min(\#\mathcal{F}_{\tau_1},\#\mathcal{F}_{\tau_2})$
	\item cost: $0$
\end{itemize}
\item $t_i^1\to t_j^2$, $t_i^1\in\mathcal{F}_{\tau_1}$, $t_j^2\in\mathcal{F}_{\tau_2}$
\begin{itemize}
	\item capacity: $1$
	\item cost: $\delta(t_i^1,t_j^2)$
\end{itemize}
\item $t_i^1\to \emptyset_{\tau_2}$, $t_i^1\in\mathcal{F}_{\tau_1}$
\begin{itemize}
	\item capacity: $1$
	\item cost: $\delta(t_i^1,\emptyset)=\#t_i^1$
\end{itemize}
\item $\emptyset_{\tau_1}\to t_j^2$, $t_j^2\in\mathcal{F}_{\tau_2}$
\begin{itemize}
	\item capacity: $1$
	\item cost: $\delta(\emptyset,t_j^2)=\#t_j^2$
\end{itemize}
\item $t_j^2\to\text{sink}$, $t_j^2\in\mathcal{F}_{\tau_2}$
\begin{itemize}
	\item capacity: $1$
	\item cost: $0$
\end{itemize}
\item $\emptyset_{\tau_2}\to\text{sink}$
\begin{itemize}
	\item capacity: $\#\mathcal{F}_{\tau_2}-\min(\#\mathcal{F}_{\tau_1},\#\mathcal{F}_{\tau_2})$
	\item cost: $0$
\end{itemize}
\end{itemize}
We obtain a network $G$ augmented with integer capa\-cities and nonnegative costs. A representation of $G$ is given in Figure~\ref{fig:mincostproblem}. By construction and as explained in \cite[Lemma 8]{Zhang:1996:CED}, one has $C(G)=\delta(\tau_1,\tau_2)$ where $C(G)$ denotes the cost of the minimum cost maximum flow on $G$. As a consequence, $\delta(\tau_1,\tau_2)$ may directly be computed from a minimum cost maximum flow algorithm presented for example in \cite[8.4 Minimum cost flows]{Tarjan:1983}. The related complexity is given in Proposition~\ref{prop:delta:complexity}.

\begin{figure}[h]
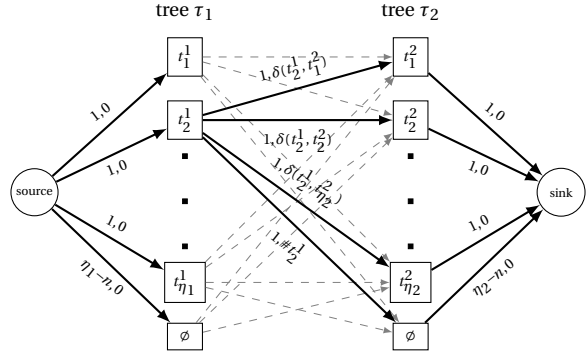

\centering
\costflowgraph
\caption{Reduction of the computation of edit distance $\delta(\tau_1,\tau_2)$ presented in Proposition~\ref{prop:calcul:ed} and in Fi\-gu\-re~\ref{fig:calcul:ed} to the minimum cost flow problem. Each edge is augmented with two labels separated by a comma: its capacity (left) and its cost (right). For the sake of simplicity, $\eta_i = \#\mathcal{F}_{\tau_i}$ and $n = \min(\eta_1,\eta_2)$.}
\label{fig:mincostproblem}
\end{figure}

\subsection{Complexity from trees}

In light of \cite[Theorem 8.13]{Tarjan:1983}, the complexity of finding the cost of the mi\-ni\-mum cost maximum flow on the network $G$ defined in Figure~\ref{fig:mincostproblem} may be directly obtained from its charac\-te\-ristics and is $O(N\times|f^\star|\times\log_2(n))$, where $n$, $N$ and $|f^\star|$ respectively denote the number of vertices, the number of edges and the maximum flow of $G$. It is quite obvious that
$$
N = O\big(\#\child(\roottree(\tau_1))\#\child(\roottree(\tau_2)) +\#\child(\roottree(\tau_1)) +\#\child(\roottree(\tau_2))\big) .
$$
In addition,
\begin{eqnarray*}
|f^\star|&=&O\big(\#\child(\roottree(\tau_1))+\#\child(\roottree(\tau_2))\big) \\
&=& O(\deg\,\tau_1 + \deg\,\tau_2 ) .
\end{eqnarray*}
And,
\begin{eqnarray*}
n&=&O\big(\#\child(\roottree(\tau_1))+\#\child(\roottree(\tau_2))\big) \\
&=& O(\deg\,\tau_1 + \deg\,\tau_2 ) .
\end{eqnarray*}
Thus, the total complexity to compute the recursive formula of $\delta(\tau_1,\tau_2)$ presented in Proposition~\ref{prop:calcul:ed} is
$$
O\Big( [\deg(\tau_1)+\deg(\tau_2)]\log_2(\deg(\tau_1)+\deg(\tau_2)) \times \sum_{v\in \tau_1}\sum_{w\in \tau_2} \#\text{child}(v)\times\#\text{child}(w)\Big) = O( \#\tau_1 \,\# \tau_2\, \psi(\tau_1,\tau_2) ) ,
$$
with $\psi(\tau_1,\tau_2) = (\degr\,\tau_1+\degr\,\tau_2) \log_2(\degr\,\tau_1+\degr\,\tau_2)$, which yields the expected result.

Not surprisingly, the time-complexity of computing the edit distance $\delta$ is the same as in \cite{Zhang:1996:CED} for another kind of constrained edit distance. It should be noted that this algorithm does not take into account the possible presence of redundant substructures, which should reduce the complexity. We tackle this question in the following part. 

\subsection{Complexity from DAG reductions}
\label{ss:complex:snt}

The tree-to-tree comparison problem has already been considered for quotiented trees (an adaptation of Zhang's algorithm \cite{Zhang:1996:CED} to quotiented trees is presented in \cite{ferraro:03}), but never for DAG reductions of trees. Taking into account the redundancies, the edit distance obtained from DAG reductions is expected to be computed with a time-complexity smaller than from trees.

In the sequel, $D_i$ denotes the DAG reduction of the tree $\tau_i$ of height $H_i$. With the notation of Subsection~\ref{ss:treereduction}, the vertices of $D_i$ are denoted by $(h , j)^i$, $0\leq h \leq H_i$ and $1\leq j\leq M_h^i$. In addition, the DAG $D_i$ is characterized by the array
$$ \big[ N^i( (h_1,j)^i ,(h_2,k)^i ) \big] .$$
We also define
$$ \eta_i = \sum_{{0\leq h<H_i} \atop {1\leq j\leq M_{h}^i}} N^i( (H_i,1)^i , (h,j)^i ) ,$$
which counts the number of children of $\roottree(\tau_i)$. As in Subsection~\ref{ss:mincostflow}, the computation of the edit distance $\delta(\tau_1,\tau_2)$ from the DAG reductions $D_1$ and $D_2$ reduces to a minimum cost flow problem but the network graph that we consider takes into account the number of appearances of a given pattern among the lists of subtrees of $\mathcal{F}_{\tau_1}$ and $\mathcal{F}_{\tau_2}$. We construct the network graph $G=(V,E)$ as follows. The set of vertices $V$ of $G$ is given by
$$
V = \left\{\text{source}\, ,\, \text{sink}\, ,\, \emptyset_1\, ,\, \emptyset_2\right\} \cup \bigcup_{{0\leq h <H_1} \atop {1\leq j \leq M_{h}^1}}\left\{ (h,j)^1 \right\}  \cup  \bigcup_{{0\leq l<H_2} \atop {1\leq k\leq M_{l}^2}}\left\{ (l,k)^2\right\} .
$$
The set $E$ of edges of $G$ is defined from:
\begin{itemize}
\item $\text{source}\to(h,j)^1$
\begin{itemize}
	\item	capacity: $N^1( (H_1,1)^1 , (h,j)^1 )$
	\item cost: $0$
\end{itemize}
\item $\text{source}\to \emptyset_1$
\begin{itemize}
	\item capacity: $\eta_1   -\min(\eta_1  ,\eta_2)$
	\item cost: $0$
\end{itemize}
\item	 $(h,j)^1 \to (l,k)^2$
\begin{itemize}
	\item capacity: $N^1 ( (H_1,1)^1 , (h,j)^1 )$
	\item cost: $\delta( \tree(D_1[ (h,j)^1] ) \,,\,\tree(D_2[ (l,k)^2] ))$
\end{itemize}
\item $(h,j)^1\to \emptyset_2$
\begin{itemize}
	\item capacity: $N^1 ( (H_1,1)^1 , (h,j)^1 )$
	\item cost: $\#\tree(D_1[ (h,j)^1])$
\end{itemize}
\item 	$\emptyset_1\to (l,k)^2$
\begin{itemize}
	\item capacity: $N^2 ( (H_2,1)^2 , (l,k)^2 )$
	\item cost: $\#\tree(D_2[ (l,k)^2])$
\end{itemize}
\item		$(l,k)^2\to\text{sink}$
\begin{itemize}
	\item capacity: $N^2 ( (H_2,1)^2 , (l,k)^2 )$
	\item cost: $0$
\end{itemize}
\item		$\emptyset_2\to\text{sink}$
\begin{itemize}
	\item capacity: $\eta_2-\min(\eta_1,\eta_2)$
	\item cost: $0$
\end{itemize}
\end{itemize}

As in Subsection~\ref{ss:mincostflow}, the graph $G$ has integer capacities and nonnegative costs on its edges. By cons\-truction, the cost $C(G)$ of the minimum cost maximum flow on the graph $G$ is equal to the expected edit distance
$$\delta( \tau_1 , \tau_2 ) = \delta\big( \tree( D_1[(H_1,1)])\,,\,\tree(D_2[(H_2,1)])\big).$$
The characteristics $n$ (number of vertices), $|f^\star|$ (maximum flow) and $N$ (number of edges) of the network $G$ satisfy:
\begin{itemize}
\item $N=O(\deg\,D_1\times\deg\,D_2+\deg\,D_1+\deg\,D_2)$;
\item $|f^\star|=O(\deg\,\tau_1+\deg\,\tau_2)$;
\item $n=O(\deg\,D_1+\deg\,D_2) = O(\deg\,\tau_1+\deg\,\tau_2)$.
\end{itemize}
Summing over the vertices of $D_1$ and $D_2$ (number of terms bounded by $\#D_1\times\#D_2$), together with \cite[Theorem 8.13]{Tarjan:1983}, states the result.

\section{Proof of Proposition~\ref{prop1}}
\label{ss:proof:prop1}

Using the notation of Subsection~\ref{ss:sntrees}, a self-nested tree of height $H$ is represented by a linear DAG with $H+1$ vertices numbered from $0$ (bottom, leaves of the tree) to $H$ (top, root of the tree) in such a way that there exists a path $H\to\dots\to1\to0$. One recalls that this graph is augmented with integer-valued label $N(i,j)$ on edge $i\to j$ for any $i>j$ with the constraint $N(i,i-1)>0$. In this context, the outdegree of a self-nested tree is less than $d$ if and only if, for any $i$,
$$\sum_{j=0}^{i-1} N(i,j) \leq d.$$
We propose to write $n_{i,i-1} = N(i,i-1)-1$ and, for $j\leq i-2$, $n_{i,j}=N(i,j)$. As a consequence, all the labels are parametrized by the $n_{i,j}$'s which satisfy, for any $i>j$, $n_{i,j}\geq0$ and, for any $i\geq1$,
$$ \sum_{j=0}^{i-1} n_{i,j}\leq d-1.$$
Thus, the number of self-nested trees of height $H$ is obtained as
\begin{equation*}
\label{cardTh}
\#\mathbb{T}^{sn}_{=H,\leq d} = \prod_{i=1}^{H}\#\Bigg\{n_{i,j}\geq0~:~\sum_{j=0}^{i-1} n_{i,j}\leq d-1\Bigg\}.
\end{equation*}
Furthermore, the set under the product sign is only the regular discrete simplex of dimension $i$ having $d$ points on an edge. The cardinality of this set has been studied by Costello in \cite{1054599}. Thus, by virtue of \cite[Theorem 2]{1054599}, one has
\begin{equation*}
\label{cardSimplex}
\#\Bigg\{n_{i,j}\geq0~:~\sum_{j=0}^{i-1} n_{i,j}\leq d-1\Bigg\} = {\binom{d+i-1}{i}} ,
\end{equation*}
which yields the expected result via a change of index.


\section{Proof of Proposition~\ref{prop2}}
\label{ss:proof:prop2}

\subsection{Asymptotics for self-nested trees}

Substituting the binomial coefficients by their value in the expression of $\#\mathbb{T}_{=H,\leq d}^{sn}$ stated in Proposition \ref{prop1}, we get
$$ \#\mathbb{T}_{=H,\leq d}^{sn} = \Gamma(d)^{-H} \prod_{i=1}^H \frac{\Gamma(d+H-i+1)}{\Gamma(H-i+2)} $$
where $\Gamma$ denotes the Euler function such that $\Gamma(n+1)=n!$ for any integer $n$. As a consequence,
\begin{eqnarray}
\log\,\#\mathbb{T}_{=H,\leq d}^{sn} &=& -H\log\,\Gamma(d) + \sum_{\substack{1\leq i\leq H\\2\leq k\leq d}} \log(H-i+k) \nonumber \\
&=& -H\log\,\Gamma(d) + \sum_{\substack{0\leq j\leq {H-1}\\2\leq k\leq d}} \log(j+k), \label{eq:cardTh:1}
\end{eqnarray}
by substituting $H-i$ by $j$. First, according to Stirling's approximation, we have
\begin{equation}\label{eq:cardTh:2}
-H\log\,\Gamma(d) \sim -Hd\log\,d .
\end{equation}
Now, we focus on the second term. In order to simplify, we are looking for an equivalent of the same double sum but indexed on $1\leq j\leq H$ and $1\leq k\leq d$. We have
\begin{eqnarray}
\sum_{j=1}^H \sum_{k=1}^d \log(j+k)& = & \sum_{j=1}^H \sum_{k=1}^d \int_1^{j+k}\frac{\text{d}x}{x}  \nonumber \\
&=& \sum_{j=1}^H \left[\sum_{l=0}^{d-1}(d-l)\int_{j+l}^{j+l+1}\frac{\text{d}x}{x} + \int_1^j \frac{\text{d}x}{x}\right] \nonumber \\
&=&\sum_{l=0}^{d-1} (d-l) \left[\sum_{j=1}^H \int_{j+l}^{j+l+1} \frac{\text{d}x}{x} + \log\,j\right] \nonumber \\
&=&\sum_{l=0}^{d-1} (d-l) \int_{l+1}^{l+H+1}\frac{\text{d}x}{x} + d\sum_{j=1}^H \log\,j \nonumber \\
&=&\sum_{l=0}^{d-1} (d-l)\log\left(1+\frac{H}{l+1}\right) + d\sum_{j=1}^H \log\,j . \label{eq:cardTh:3}
\end{eqnarray}
As usually, we find an equivalent of this term by using an integral comparison test. We establish by a conscientious calculus that
\begin{equation}
\sum_{l=0}^{d-1} (d-l)\log\left(1+\frac{H}{l+1}\right) + d\sum_{j=1}^H \log\,j \sim\,\frac{(d+H)^2}{2}\log(d+H)-\frac{H^2}{2}\log\,H -\frac{d^2}{2}\log\,d + R(d,H), \label{eq:cardTh:4}
\end{equation}
where the rest $R(d,H)$ is neglectable with respect to the other terms and to $Hd\log\,d$. Let us remark that the expression of the equivalent is symmetric in $H$ and $d$ as expected. Finally, \eqref{eq:cardTh:1}, \eqref{eq:cardTh:2}, \eqref{eq:cardTh:3} and \eqref{eq:cardTh:4} show the result.

\subsection{Asymptotics for unordered trees}

Roughly speaking, an unordered tree with maximal height $H$ and maximal outdegree $d$ may be obtained by adding at most $d$ trees of height less than $H-1$ to an isolated root. More precisely, one has to choose $d$ elements with repetitions among the set $\mathbb{T}_{\leq H-1,\leq d}\cup\{\bullet\}\cup\{\emptyset\}$ and add them to the list of children (initially empty) of a same vertex. It should be noted that no subtree is added when $\emptyset$ is picked.

One obtains either an isolated root (if and only if one draws $d$ times the symbol $\emptyset$), or a tree with maximal height $H$. As a consequence, one has the formula,
$$
\#\left[\mathbb{T}_{\leq H,\leq d}\cup\{\bullet\}\right] = {\binom{\#\left[\mathbb{T}_{\leq H-1,\leq d}\cup\{\bullet\}\cup\{\emptyset\}\right]+d-1}{d}},
$$
which shows that
$$\#\mathbb{T}_{\leq H,\leq d} = u_H(d)-1,$$
with $u_0(d) = 1$ and 
$$ u_H(d)={\binom{u_{H-1}(d)+d}{d}} .$$
The sequel of the proof is based on the classical bounds on binomial coefficients,
$$\left(\frac{n\times\e}{k}\right)^k \geq \binom{n}{k} \geq \left(\frac{n}{k}\right)^k .$$
We have $u_1(d) = \binom{1+d}{d} = d+1$ and 
\begin{eqnarray*}
u_2(d) &=& \binom{u_1(d)+d}{d} \\
&=& \binom{2d+1}{d}\\
&\geq& \left( \frac{2d+1}{d}\right)^d \\
&=& \left(2+ \frac{1}{d}\right)^d.
\end{eqnarray*}
The lower bound is obtained by induction on $H>2$: assuming that
$$u_H(d) \geq \frac{\left(2+\frac{1}{d}\right)^{d^{H-1}}}
{d^{{\frac{d^{H-1} - 1}{d-1} - 1}}},$$ we have
\begin{eqnarray*}
u_{H+1}(d) &=& \binom{u_H(d)+d}{d}\\
&\geq& \left(\frac{u_H(d)}{d}+1\right)^d \\
&\geq& \left(\frac{u_H(d)}{d}\right)^d \\
&\geq& \left(\frac{\left(2+\frac{1}{d}\right)^{d^{H-1}}}{d^{{\frac{d^{H-1} - 1}{d-1}}}}\right)^d
\end{eqnarray*}
by the induction hypothesis. Using
$$
d\left(\frac{d^{H-1} - 1}{d-1}\right) = \frac{d^H - 1}{d-1} - 1,
$$
we obtain
$$
u_{H+1}(d) \geq \frac{\left(2+\frac{1}{d}\right)^{d^H}}
{d^{{\frac{d^{H} - 1}{d-1}-1}}}.
$$
Moreover,
$$u_2(d) = \binom{2d+1}{d}
\leq \left( \frac{2d+1}{d}\e\right)^d \leq (3\e)^d.
$$
The upper bound is also obtained by induction on $H \geq 2$: assuming that
$$u_H(d) \leq 3^{d^{H-1}}\e^{\frac{d^H - 1}{d-1} - 1},$$ we obtain
\begin{eqnarray*}
u_{H+1}(d) &=& \binom{u_H(d)+d}{d} \\
&\leq& \left(\left(\frac{u_H(d)}{d}+1\right)\e\right)^d \\
&\leq& \left(\left(\frac{3^{d^{H-1}}\e^{\frac{d^H - 1}{d-1} - 1}}{d} + 1\right)\e\right)^d
\end{eqnarray*}
by the induction hypothesis. Using the inequality
$$
\frac{k^x}{x} + 1 \leq k^x, 
$$
satisfied whenever $k$ and $x$ are both greater than the critical value $1.693\dots$ (obtained by numerical methods), we obtain
\begin{eqnarray*}
u_{H+1}(d) &\leq&  \left(3^{d^{H-1}}\e^{\frac{d^H - 1}{d-1} - 1}\e\right)^d \\
&= & 3^{d^H}\e^{\frac{d^{H+1} - 1}{d-1} - 1}.
\end{eqnarray*}
This shows the expected result.


\section{Proof of Proposition~\ref{prop:complav}}

Computing all the multiplicities $\mu((h_1,i))$ can be done in one traversal of all the edges of $D$ via the recursive formula \eqref{eq:multipli}. By definition of $\widehat{D}$, for each couple $(h_1,h_2)$, computing $\widehat{N}(h_1,h_2)$ requires to traverse all the edges $(h_1,i)\to(h_2,j)$ with no overlap. Finally, all the edges of $D$ have been traversed once, which states the complexity.


\section{Proof of Proposition~\ref{prop:wc}}
\label{s:proof:wc}

We begin this proof with trees of height $2$, and we shall state in two steps the expected result.

First of all, let us remark that the DAG of any tree of $\mathbb{T}_{2,\leq m}$ is of the form\!\!\minidagf. Nevertheless, leaves attached to the root do not impact the self-nestedness of the tree and deletes some degrees of freedom in our research of the worst case. As a consequence, we only consider DAGs of the form\!\!\minidag with $n$ intermediate vertices (that is to say $n$ different subtrees of height $1$) labeled from $I_1$ to $I_n$, $n\leq d$. Of course, $n=1$ ensures that the corresponding tree is self-nested: we exclude this case. Let $p_k$ ($l_k$, respectively) denote the number of appearances (the number of leaves, respectively) of $I_k$, for $1\leq k\leq n$.

We shall investigate the worst case for a given value of $n$. First, it should be noted that if an operation is optimal for an equivalence class $I_k$, it is also optimal for all the subtrees of this class. In addition, there are only two possible scripts to transform $I_k$: either one deletes all the leaves of $I_k$ (with a cost $p_k l_k$), or one adds or deletes some leaves to transform $I_k$ into a given subtree of height $1$ with, say, $x$ leaves (with a cost $p_k|l_k-x|$). As a consequence, the total editing cost (to transform the initial tree into a self-nested tree in which trees of height $1$ have $x$ leaves) is given by
$$C_2=\sum_{k\in A} p_k l_k + \sum_{k\notin A} p_k|l_k-x| ,$$
where $A$ denotes the set of indices $k$ for which one deletes all the leaves of $I_k$.

The worst case has the maximum entropy and thus a uniform repartition of its leaves in the tree. For the sake of clarity, one assumes in the sequel that $d$ is even and $n$ divides $d$. The explicit solution of the problem is thus $p_k=\frac{d}{n}$, $l_k=\frac{kd}{n}$, $x=\frac{d}{2}$ and $A=\emptyset$. The remarkable fact is that the corresponding cost is given by
\begin{eqnarray*}
C_2	&=&	\frac{d}{n}\sum_{k=1}^n\left|\frac{d}{2}-\frac{kd}{n}\right|\\
	&=&	\frac{2d}{n} \sum_{k=1}^{\frac{n}{2}-1} \left(\frac{d}{2}-\frac{kd}{n}\right)\\
	&=&	\frac{d^2}{4}.
\end{eqnarray*}
This means that the worst case may be obtained from any value of $n$ whenever it divides $d$. Actually, the case $n$ does not divide $d$ leads to a worst case better than when $n$ divides $d$. One concludes that one of the worst cases is obtained from $n=2$, $p_1=p_2=\frac{d}{2}$, $l_1=\frac{d}{2}$, $l_2=d$ and $C_2=\frac{d^2}{4}$. When $d$ is an odd integer, one observes the same phenomenon: the worst case is obtained from $n=2$, $p_1=\left\lceil\frac{d}{2}\right\rceil$, $p_2=\left\lfloor\frac{d}{2}\right\rfloor$, $l_1=\left\lfloor\frac{d}{2}\right\rfloor$, $l_2=d$ and $C_2=\left\lfloor\frac{d}{2}\right\rfloor\times \left\lceil\frac{d}{2}\right\rceil$. This yields the expected result for any integer $d$.

We shall use the preceding idea to show the result for any height $H$. Among trees of height at most $H$, it is quite obvious that the worst case appears in trees of height $H$. We assume that there are $n$ different patterns $I_1,\dots,I_n$ appearing $p_1,\dots,p_n$ times under the root. The cost of editing operations (adding or deleting leaves) at distance $h$ to the root is in the worst case $p_k \, d^{h-1}$. As a consequence, at least for $d$ large enough, $\text{height}(I_k)=H-1$ and the only difference with the other patterns is on the fringe: all the vertices of $I_k$ have $d$ children except vertices at height $H-2$ that have $l_k$ leaves. If $A$ denotes the set of indices $k$ for which one deletes all the leaves of $I_k$, the editing cost to transform the tree into the self-nested tree in which subtrees of height $1$ have $x$ leaves is given by
$$C_H=d^{H-2}\left[\sum_{k\in A} p_k l_k + \sum_{k\notin A} p_k |l_k-x|\right].$$
In light of the previous reasoning, this states the expected result.

\end{document}

%% file: tikz_costflow.tex
\newcommand{\costflowgraph}{
\begin{tikzpicture}[xscale=0.10,yscale=0.12]
\tikzstyle{fleche}=[->,>=latex,thick,style={sloped,anchor=south,auto=false}]
\tikzstyle{flechegr}=[->,>=latex,very thin,dashed,color=gray]

\tikzstyle{noeud}=[fill=white,draw]
\tikzstyle{noeudpoint}=[fill=black,draw,scale=0.3]
\tikzstyle{noeudrond}=[fill=white,draw,circle]
\tikzstyle{numberr}=[draw=none,fill=none,scale=1]

\tikzstyle{etiquette}=[midway]

\def\shift{15}
\def\grandshift{35}

\node[noeudrond] (source) at (-\grandshift , 0) {\tiny \!\!\!\! source \!\!\!\!};
\node[noeudrond] (puits) at (\grandshift , 0) {\tiny $\,$sink$\,$};

\node[numberr] (tree1) at (-\shift,20) {\footnotesize tree $\tau_1$};
\node[numberr] (tree2) at (\shift,20) {\footnotesize tree $\tau_2$};

\node[noeud] (t1) at (-\shift , 15) {\tiny $\,t_{1~}^1$};
\node[noeud] (tau1) at (\shift , 15) {\tiny $\,t_{1~}^2$};

\node[noeud] (t2) at (-\shift , 8) {{\tiny $\,t_{2~}^1$}};
\node[noeud] (tau2) at (\shift , 8) {\tiny $\,t_{2~}^2$};

\node[noeudpoint] (l1) at (-\shift,4) {};
\node[noeudpoint] (l2) at (-\shift,-1) {};
\node[noeudpoint] (l3) at (-\shift,-6) {};
\node[noeudpoint] (l11) at (\shift,4) {};
\node[noeudpoint] (l12) at (\shift,-1) {};
\node[noeudpoint] (l13) at (\shift,-6) {};

\node[noeud] (tnt) at (-\shift , -10) {\tiny $t_{\eta_1}^1$};
\node[noeud] (tauntau) at (\shift , -10) {\tiny $t_{\eta_2}^2$};

\node[noeud] (e1) at (-\shift , -16) {\tiny $~\,\emptyset\,~$};
\node[noeud] (e2) at (\shift , -16) {\tiny $~\,\emptyset\,~$};

\draw[flechegr] (t1) -- (tau1) {};
\draw[flechegr] (t1) -- (tau2) {};
\draw[flechegr] (t1) -- (tauntau) {};
\draw[flechegr] (t1) -- (e2) {};
\draw[flechegr] (tnt) -- (tau1) {};
\draw[flechegr] (tnt) -- (tau2) {};
\draw[flechegr] (tnt) -- (tauntau) {};
\draw[flechegr] (tnt) -- (e2) {};<
\draw[flechegr] (e1) -- (tau1) {};
\draw[flechegr] (e1) -- (tau2) {};
\draw[flechegr] (e1) -- (tauntau) {};

\draw[fleche] (source) -- node{\tiny$1,0$} (t1) {};
\draw[fleche] (source) -- node[below]{\tiny$1,0$}  (t2) {};
\draw[fleche] (source) -- node{\tiny$1,0$} (tnt) {};
\draw[fleche] (source) -- node[below]{\tiny$\eta_1\!\!-\!\!n,0$}  (e1) {};

\draw[fleche] (t2) -- node{\tiny $1,\delta(t_{2}^1,t_{1}^2)$} (tau1) {};
\draw[fleche] (t2) -- node[below]{\tiny $1,\delta(t_{2}^1,t_{2}^2)$} (tau2) {};
\draw[fleche] (t2) -- node{\tiny $1,\delta(t_{2}^1,t_{\eta_2}^2)$} (tauntau) {};
\draw[fleche] (t2) -- node[below]{\tiny $1,\#t_{2}^1$} (e2) {};

\draw[fleche] (tau1) -- node{\tiny$1,0$} (puits) {};
\draw[fleche] (tau2) -- node[below]{\tiny$1,0$}  (puits) {};
\draw[fleche] (tauntau) -- node{\tiny$1,0$} (puits) {};
\draw[fleche] (e2) -- node[below]{\tiny$\eta_2\!\!-\!\!n,0$}  (puits) {};
\end{tikzpicture}}

%% file: tikz_thetaT.tex
\newcommand{\dagTcal}{
\begin{tikzpicture}[xscale=0.9,yscale=0.6]
\tikzstyle{fleche}=[->,>=latex,thin]
\tikzstyle{dfleche}=[<->,>=latex,thin,black] 
\tikzstyle{noeud}=[fill=black,circle,draw,scale=0.3]
\tikzstyle{point}=[fill=black,circle,draw,scale=0.1]
\tikzstyle{etiquette}=[midway,right]
\tikzstyle{eti}=[midway,left]
\tikzstyle{etiii}=[midway,right]
\tikzstyle{nom}=[]

\node[noeud] (A) at (0,0) {};
\node[noeud] (B) at (0,-1) {};
\node[noeud] (C) at (0,-1.8) {};
\node[point] (Z) at (0,-2.4) {};
\node[point] (ZZ) at (0,-3.6) {};
\node[point] (ZZZ) at (0,-3) {};
\node[noeud] (D) at (0,-4.2) {};
\node[noeud] (E) at (0,-5) {};
\node[noeud] (F) at (0,-6) {};

\draw[dfleche] (0.7,-1.5) -- (0.7,-4.5) node[etiii] {\tiny $H\!-\!2$};
\draw[fleche] (A)--(B) node[etiquette] {\footnotesize$d$};
\draw[fleche] (B)--(C) node[etiquette] {\footnotesize$d$};
\draw[fleche] (D)--(E) node[etiquette] {\footnotesize$d$};
\draw[fleche] (E)--(F) node[etiquette] {\footnotesize$\left\lfloor\frac{d}{2}\right\rfloor$};

\def\shift{-4.5}

\tikzstyle{fleche}=[->,>=latex,thin]
\tikzstyle{dfleche}=[<->,>=latex,thin,black]
\tikzstyle{noeud}=[fill=black,circle,draw,scale=0.3]
\tikzstyle{point}=[fill=black,circle,draw,scale=0.1]
\tikzstyle{etiquette}=[midway,right]
\tikzstyle{eti}=[midway,left]
\tikzstyle{nom}=[]

\node[noeud] (A) at (\shift+0,0) {};

\node[noeud] (B) at (\shift-1,-1) {};
\node[noeud] (C) at (\shift-1,-1.8) {};
\node[point] (Z) at (\shift-1,-2.4) {};
\node[point] (ZZ) at (\shift-1,-3.6) {};
\node[point] (ZZZ) at (\shift-1,-3) {};
\node[noeud] (D) at (\shift-1,-4.2) {};
\node[noeud] (E) at (\shift-1,-5) {};

\node[noeud] (BB) at (\shift+1,-1) {};
\node[noeud] (CC) at (\shift+1,-1.8) {};
\node[point] (Z1) at (\shift+1,-2.4) {};
\node[point] (ZZ1) at (\shift+1,-3.6) {};
\node[point] (ZZZ1) at (\shift+1,-3) {};
\node[noeud] (DD) at (\shift+1,-4.2) {};
\node[noeud] (EE) at (\shift+1,-5) {};

\node[noeud] (F) at (\shift+0,-6) {};

\draw[dfleche] (\shift-1.5,-1.5) -- (\shift-1.5,-4.5) node[eti] {\tiny $H\!-\!2$};
\draw[fleche] (A)--(B) node[eti,,above=0.8pt,near end] {\footnotesize$\left\lceil\frac{d}{2}\right\rceil$~~};
\draw[fleche] (B)--(C) node[etiquette] {\footnotesize$d$};
\draw[fleche] (D)--(E) node[etiquette] {\footnotesize$d$};
\draw[fleche] (E)--(F) node[eti,near start,below=0.2pt] {\footnotesize$\left\lfloor\frac{d}{2}\right\rfloor$~~~};

\draw[fleche] (A)--(BB) node[etiquette,near end,above=0.8pt] {\footnotesize~~$\left\lfloor\frac{d}{2}\right\rfloor$};
\draw[fleche] (BB)--(CC) node[etiquette] {\footnotesize$d$};
\draw[fleche] (DD)--(EE) node[etiquette] {\footnotesize$d$};
\draw[fleche] (EE)--(F) node[etiquette,below=0.2pt,near start] {\footnotesize~$d$};

\end{tikzpicture}}

%% file: tikz_minidags.tex
\newcommand{\minidagf}{
\begin{tikzpicture}[xscale=0.25,yscale=0.15]
\tikzstyle{fleche}=[-,>=latex,thin]
\tikzstyle{noeud}=[fill=black,circle,draw,scale=0.15]
\tikzstyle{point}=[fill=black,circle,draw,scale=0.1]

\node[noeud] (R) at (0,1) {};
\node[noeud] (I1) at (-1,0) {};
\node[noeud] (I2) at (-0.5,0) {};
\node[noeud] (I3) at (1,0) {};
\node[noeud] (L) at (0,-1) {};

\node[point] (p1) at (-0.1,0) {};
\node[point] (p2) at (0.25,0) {};
\node[point] (p3) at (0.6,0) {};

\draw[fleche] (R)--(I1) node {};
\draw[fleche] (R)--(I2) node {};
\draw[fleche] (R)--(I3) node {};
\draw[fleche] (L)--(I1) node {};
\draw[fleche] (L)--(I2) node {};
\draw[fleche] (L)--(I3) node {};

\draw[fleche] (0,1) .. controls (-2.1,0) and (-2.1,0) .. (0,-1);

\end{tikzpicture}}

\newcommand{\minidag}{
\begin{tikzpicture}[xscale=0.25,yscale=0.15]
\tikzstyle{fleche}=[-,>=latex,thin]
\tikzstyle{noeud}=[fill=black,circle,draw,scale=0.15]
\tikzstyle{point}=[fill=black,circle,draw,scale=0.1]

\node[noeud] (R) at (0,1) {};
\node[noeud] (I1) at (-1,0) {};
\node[noeud] (I2) at (-0.5,0) {};
\node[noeud] (I3) at (1,0) {};
\node[noeud] (L) at (0,-1) {};

\node[point] (p1) at (-0.1,0) {};
\node[point] (p2) at (0.25,0) {};
\node[point] (p3) at (0.6,0) {};

\draw[fleche] (R)--(I1) node {};
\draw[fleche] (R)--(I2) node {};
\draw[fleche] (R)--(I3) node {};
\draw[fleche] (L)--(I1) node {};
\draw[fleche] (L)--(I2) node {};
\draw[fleche] (L)--(I3) node {};
\end{tikzpicture}}

%% file: tikz_algo-ed.tex
\newcommand{\algoed}{
\begin{tikzpicture}[xscale=0.3,yscale=0.5]

\tikzstyle{fleche}=[->,>=latex,thick]
\tikzstyle{noeud}=[fill=white,circle,draw,scale=0.5]
\tikzstyle{noeuddd}=[fill=white,circle,draw,scale=0.415]
\tikzstyle{noeudg}=[fill,circle,color=black!15!white,draw,scale=0.5]
\tikzstyle{flecheg}=[->,>=latex,thick,color=black!15!white,dash pattern=on 0.5mm off 0.5mm]
\tikzstyle{noeudgg}=[fill,circle,color=black!45!white,draw,scale=0.5]
\tikzstyle{flechegg}=[->,>=latex,thick,color=black!55!white,dash pattern=on 0.5mm off 0.5mm]

\tikzstyle{flecherr}=[->,>=latex,thick,color=red]
\tikzstyle{flechebl}=[->,>=latex,thick,color=blue]

\tikzstyle{noeudb}=[fill,circle,color=white,draw,scale=0.5]
\tikzstyle{flecheb}=[->,>=latex,thick,color=white]
\tikzstyle{number}=[draw=none,fill=none,scale=1.2]
\tikzstyle{numberr}=[draw=none,fill=none,scale=1]
\tikzstyle{etiquette}=[midway]

\def\DistanceInterNiveaux{1}
\def\DistanceInterFeuilles{1}

\def\NiveauA{(-0)*\DistanceInterNiveaux}
\def\NiveauB{(-1.1)*\DistanceInterNiveaux}
\def\NiveauC{(-2.2)*\DistanceInterNiveaux}
\def\InterFeuilles{(1)*\DistanceInterFeuilles}

\node[numberr] (A0) at (0,1) {\footnotesize tree $\tau_1$};
\node[noeud] (A) at ({(0)*\InterFeuilles},{\NiveauA}) {};
\node[noeud] (B1) at ({(-3)*\InterFeuilles},{\NiveauB}) {};
\node[noeud] (B2) at ({(0)*\InterFeuilles},{\NiveauB}) {};
\node[noeud] (B3) at ({(3)*\InterFeuilles},{\NiveauB}) {};

\draw[fleche] (A)--(B1) node[etiquette] {};
\draw[fleche] (A)--(B2) node[etiquette] {};
\draw[fleche] (A)--(B3) node[etiquette] {};

\draw (B1) -- (-4,-4);
\draw (B1) -- (-2,-4);
\draw (-4,-4) -- (-2,-4);

\draw (B2) -- (-1,-3);
\draw (B2) -- (1,-3);
\draw (1,-3) -- (-1,-3);

\draw (B3) -- (2,-3.5);
\draw (B3) -- (4,-3.5);
\draw (4,-3.5) -- (2,-3.5);

\def\shift{15}

\node[numberr] (A0) at (\shift+0,1) {\footnotesize tree $\tau_2$};
\node[noeud] (AA) at ({(\shift+0)*\InterFeuilles},{\NiveauA}) {};
\node[noeud] (BB1) at ({(\shift-3)*\InterFeuilles},{\NiveauB}) {};
\node[noeud] (BB2) at ({(\shift-1)*\InterFeuilles},{\NiveauB}) {};
\node[noeud] (BB3) at ({(\shift+1)*\InterFeuilles},{\NiveauB}) {};
\node[noeud] (BB4) at ({(\shift+3)*\InterFeuilles},{\NiveauB}) {};

\draw[fleche] (AA)--(BB1) node[etiquette] {};
\draw[fleche] (AA)--(BB2) node[etiquette] {};
\draw[fleche] (AA)--(BB3) node[etiquette] {};
\draw[fleche] (AA)--(BB4) node[etiquette] {};

\draw (BB1) -- (\shift-3.5,-2.5);
\draw (BB1) -- (\shift-2.5,-2.5);
\draw (\shift-3.5,-2.5) -- (\shift-2.5,-2.5);

\draw (BB2) -- (\shift-2,-4);
\draw (BB2) -- (\shift-0,-4);
\draw (\shift-2,-4) -- (\shift-0,-4);

\draw (BB3) -- (\shift+0.2,-3.5);
\draw (BB3) -- (\shift+1.8,-3.5);
\draw (\shift+0.2,-3.5) -- (\shift+1.8,-3.5);

\draw (BB4) -- (\shift+2.5,-2);
\draw (BB4) -- (\shift+3.5,-2);
\draw (\shift+2.5,-2) -- (\shift+3.5,-2);

\node[number] (ES) at (7.5,-4) {$\emptyset$};

\draw[flecherr] (B3) edge[out=20,in=160] (BB3) node[etiquette] {};
\draw[flecherr] (B2) edge[out=-20,in=-160] (BB1) node[etiquette] {};
\draw[flecherr] (B1) edge[out=-20,in=-150] (BB2) node[etiquette] {};

\draw[flecherr] (ES) edge[out=0,in=-130] (BB4) node[etiquette] {};

\end{tikzpicture}}

%% file: tikz_exdag.tex
\newcommand{\treeTA}{
\begin{tikzpicture}[xscale=0.15,yscale=0.15]
\tikzstyle{fleche}=[->,>=latex,thick]
\tikzstyle{noeud}=[fill=white,circle,draw]


\tikzstyle{noeudn}=[fill=black,circle,draw,scale=0.5]
\tikzstyle{noeudb}=[fill=blue,circle,draw,scale=0.5]
\tikzstyle{noeudr}=[fill=red,circle,draw,scale=0.5]
\tikzstyle{noeudg}=[fill=green!70!black,circle,draw,scale=0.5]

\tikzstyle{etiquette}=[midway]

\def\DistanceInterNiveaux{5}
\def\DistanceInterFeuilles{5}

\def\NiveauA{(-0)*\DistanceInterNiveaux}
\def\NiveauB{(-1.1)*\DistanceInterNiveaux}
\def\NiveauC{(-2.2)*\DistanceInterNiveaux}
\def\InterFeuilles{(1)*\DistanceInterFeuilles}

\node[noeudn] (A) at ({(0)*\InterFeuilles},{\NiveauA}) {};

\node[noeudb] (B1) at ({(-2)*\InterFeuilles},{\NiveauB}) {};
\node[noeudr] (B2) at ({(-0.25)*\InterFeuilles},{\NiveauB}) {};
\node[noeudr] (B3) at ({(1)*\InterFeuilles},{\NiveauB}) {};
\node[noeudg] (B4) at ({(2.2)*\InterFeuilles},{\NiveauB}) {};

\node[noeudg] (C11) at ({(-2.5)*\InterFeuilles},{\NiveauC}) {};
\node[noeudg] (C12) at ({(-2)*\InterFeuilles},{\NiveauC}) {};
\node[noeudg] (C13) at ({(-1.5)*\InterFeuilles},{\NiveauC}) {};

\node[noeudg] (C21) at ({(-0.5)*\InterFeuilles},{\NiveauC}) {};
\node[noeudg] (C22) at ({(0)*\InterFeuilles},{\NiveauC}) {};

\node[noeudg] (C31) at ({(0.75)*\InterFeuilles},{\NiveauC}) {};
\node[noeudg] (C32) at ({(1.25)*\InterFeuilles},{\NiveauC}) {};

\draw[fleche] (A)--(B1) node[etiquette] {};
\draw[fleche] (A)--(B2) node[etiquette] {};
\draw[fleche] (A)--(B3) node[etiquette] {};
\draw[fleche] (A)--(B4) node[etiquette] {};
\draw[fleche] (B1)--(C11) node[etiquette] {};
\draw[fleche] (B1)--(C12) node[etiquette] {};
\draw[fleche] (B1)--(C13) node[etiquette] {};
\draw[fleche] (B2)--(C21) node[etiquette] {};
\draw[fleche] (B2)--(C22) node[etiquette] {};
\draw[fleche] (B3)--(C31) node[etiquette] {};
\draw[fleche] (B3)--(C32) node[etiquette] {};


\tikzstyle{noeud}=[fill=black,draw,circle,scale=0.5]
\tikzstyle{noeudrouge}=[fill=red,draw,circle,scale=0.5]
\tikzstyle{noeudbleu}=[fill=blue,draw,circle,scale=0.5]
\tikzstyle{noeudvert}=[fill=green!70!black,draw,circle,scale=0.5]

\tikzstyle{etiquette}=[midway]

\def\shift{8}

\node[noeudvert] (DAG4) at ({\shift+(4)*\InterFeuilles},{\NiveauC}) {};
\node[noeudbleu] (DAG3) at ({\shift+(3.2)*\InterFeuilles},{\NiveauB}) {};
\node[noeudrouge] (DAG2) at ({\shift+(4)*\InterFeuilles},{\NiveauB}) {};
\node[noeud] (DAG1) at ({\shift+(4)*\InterFeuilles},{\NiveauA}) {};

\draw[fleche] (DAG3)--(DAG4) node[etiquette] {\footnotesize$3~~~\,$};
\draw[fleche] (DAG2)--(DAG4) node[etiquette] {\footnotesize$\quad2$};
\draw[fleche] (DAG1)--(DAG3) node[etiquette] {\footnotesize$1~~~\,$};
\draw[fleche] (DAG1)--(DAG2) node[etiquette] {\footnotesize$\quad2$};
\path[fleche] (DAG1) edge[bend left=60] node {\footnotesize$\quad1$} (DAG4);

\end{tikzpicture}}

\newcommand{\treeTB}{
\begin{tikzpicture}[xscale=0.15,yscale=0.15]
\tikzstyle{fleche}=[->,>=latex,thick]
\tikzstyle{noeud}=[fill=white,circle,draw]

\tikzstyle{noeudn}=[fill=black,circle,draw,scale=0.5]
\tikzstyle{noeudb}=[fill=blue,circle,draw,scale=0.5]
\tikzstyle{noeudr}=[fill=red,circle,draw,scale=0.5]
\tikzstyle{noeudg}=[fill=green!70!black,circle,draw,scale=0.5]

\tikzstyle{etiquette}=[midway]

\def\DistanceInterNiveaux{5}
\def\DistanceInterFeuilles{5}

\def\NiveauA{(-0)*\DistanceInterNiveaux}
\def\NiveauB{(-1.1)*\DistanceInterNiveaux}
\def\NiveauC{(-2.2)*\DistanceInterNiveaux}
\def\InterFeuilles{(1)*\DistanceInterFeuilles}

\node[noeudn] (A) at ({(0)*\InterFeuilles},{\NiveauA}) {};

\node[noeudr] (B1) at ({(-1.5)*\InterFeuilles},{\NiveauB}) {};
\node[noeudr] (B2) at ({(0)*\InterFeuilles},{\NiveauB}) {};
\node[noeudr] (B3) at ({(1.5)*\InterFeuilles},{\NiveauB}) {};
\node[noeudg] (B4) at ({(2.5)*\InterFeuilles},{\NiveauB}) {};

\node[noeudg] (C11) at ({(-1.8)*\InterFeuilles},{\NiveauC}) {};
\node[noeudg] (C13) at ({(-1.2)*\InterFeuilles},{\NiveauC}) {};

\node[noeudg] (C21) at ({(-0.3)*\InterFeuilles},{\NiveauC}) {};
\node[noeudg] (C23) at ({(0.3)*\InterFeuilles},{\NiveauC}) {};

\node[noeudg] (C31) at ({(1.8)*\InterFeuilles},{\NiveauC}) {};
\node[noeudg] (C33) at ({(1.2)*\InterFeuilles},{\NiveauC}) {};

\draw[fleche] (A)--(B1) node[etiquette] {};
\draw[fleche] (A)--(B2) node[etiquette] {};
\draw[fleche] (A)--(B3) node[etiquette] {};
\draw[fleche] (A)--(B4) node[etiquette] {};
\draw[fleche] (B1)--(C11) node[etiquette] {};
\draw[fleche] (B1)--(C13) node[etiquette] {};
\draw[fleche] (B2)--(C21) node[etiquette] {};
\draw[fleche] (B2)--(C23) node[etiquette] {};
\draw[fleche] (B3)--(C31) node[etiquette] {};
\draw[fleche] (B3)--(C33) node[etiquette] {};

\tikzstyle{fleche}=[->,>=latex,thick]

\tikzstyle{noeud}=[fill=black,draw,circle,scale=0.5]
\tikzstyle{noeudrouge}=[fill=red,draw,circle,scale=0.5]
\tikzstyle{noeudbleu}=[fill=blue,draw,circle,scale=0.5]
\tikzstyle{noeudvert}=[fill=green!70!black,draw,circle,scale=0.5]

\tikzstyle{etiquette}=[midway]

\def\shift{8}

\node[noeudvert] (DAG3) at ({\shift+(4)*\InterFeuilles},{\NiveauC}) {};
\node[noeudrouge] (DAG2) at ({\shift+(4)*\InterFeuilles},{\NiveauB}) {};
\node[noeud] (DAG1) at ({\shift+(4)*\InterFeuilles},{\NiveauA}) {};

\draw[fleche] (DAG2)--(DAG3) node[etiquette] {\footnotesize$2\quad$};
\draw[fleche] (DAG1)--(DAG2) node[etiquette] {\footnotesize$3\quad$};
\path[fleche] (DAG1) edge[bend left=60] node {\footnotesize$\quad1$} (DAG3);

\end{tikzpicture}}

\newcommand{\treeNot}{
\begin{tikzpicture}[xscale=0.15,yscale=0.15]
\tikzstyle{fleche}=[->,>=latex,thick]
\tikzstyle{fleche2}=[<->,>=latex,black!40!white]
\tikzstyle{noeud}=[fill=white,circle,draw,scale=0.5]
\tikzstyle{point}=[fill=none,circle,draw=none,scale=0.1]
\tikzstyle{pointtt}=[fill=black,circle,draw,scale=0.1]
\tikzstyle{number}=[fill=none,scale=0.5]

\tikzstyle{noeudnoir}=[fill=black!45!white,draw,circle,scale=0.5]
\tikzstyle{noeudrouge}=[fill=red!45!white,draw,circle,scale=0.5]
\tikzstyle{noeudbleu}=[fill=blue!45!white,draw,circle,scale=0.5]
\tikzstyle{noeudvert}=[fill=green!75!black!45!white,draw,circle,scale=0.5]

\tikzstyle{etiquette}=[midway]

\def\DistanceInterNiveaux{5}
\def\DistanceInterFeuilles{5}

\def\NiveauA{(-0)*\DistanceInterNiveaux}
\def\NiveauB{(-1.1)*\DistanceInterNiveaux}
\def\NiveauC{(-2.2)*\DistanceInterNiveaux}
\def\InterFeuilles{(1)*\DistanceInterFeuilles}

\node[noeudnoir] (A) at ({(0)*\InterFeuilles},{\NiveauA}) {};

\node[noeudrouge] (B1) at ({(-1)*\InterFeuilles},{\NiveauB}) {};
\node[point] (P1) at ({(-0.6)*\InterFeuilles},{\NiveauB-0.6}) {};
\node[pointtt] (PPP1) at ({(-0.4)*\InterFeuilles},{\NiveauB}) {};
\node[pointtt] (PPP2) at ({(0)*\InterFeuilles},{\NiveauB}) {};
\node[pointtt] (PPP3) at ({(0.4)*\InterFeuilles},{\NiveauB}) {};
\node[number] at ({(0)*\InterFeuilles},{-6.5-0.4}) {\footnotesize$n_{2,1}$};
\node[point] (P11) at ({(0.6)*\InterFeuilles},{\NiveauB-0.6}) {};
\node[noeudrouge] (B2) at ({(1)*\InterFeuilles},{\NiveauB}) {};

\node[noeudvert] (B3) at ({(2.3)*\InterFeuilles},{\NiveauB}) {};
\node[point] (P2) at ({(2.6)*\InterFeuilles},{\NiveauB-0.6}) {};
\node[pointtt] (PPPT1) at ({(2.7)*\InterFeuilles},{\NiveauB}) {};
\node[pointtt] (PPPT2) at ({(3.4)*\InterFeuilles},{\NiveauB}) {};
\node[pointtt] (PPPT3) at ({(3.05)*\InterFeuilles},{\NiveauB}) {};
\node[number] at ({(3.1)*\InterFeuilles},{-6.5-0.4}) {\footnotesize$n_{2,0}$};
\node[point] (P22) at ({(3.5)*\InterFeuilles},{\NiveauB-0.6}) {};
\node[noeudvert] (B4) at ({(4)*\InterFeuilles},{\NiveauB}) {};

\node[noeudvert] (C11) at ({(-1.5)*\InterFeuilles},{\NiveauC}) {};

\node[pointtt] (PPPZ1) at ({(-1.2)*\InterFeuilles},{\NiveauC}) {};
\node[pointtt] (PPPZ2) at ({(-1)*\InterFeuilles},{\NiveauC}) {};
\node[pointtt] (PPPZ3) at ({(-0.8)*\InterFeuilles},{\NiveauC}) {};

\node[point] (P3) at ({(-1.3)*\InterFeuilles},{\NiveauC-0.6}) {};
\node[number] at ({(-1)*\InterFeuilles},{-12-0.5}) {\footnotesize$n_{1,0}$};
\node[point] (P33) at ({(-0.7)*\InterFeuilles},{\NiveauC-0.6}) {};
\node[noeudvert] (C12) at ({(-0.5)*\InterFeuilles},{\NiveauC}) {};

\node[noeudvert] (C21) at ({(0.5)*\InterFeuilles},{\NiveauC}) {};

\node[pointtt] (PPPY1) at ({(0.8)*\InterFeuilles},{\NiveauC}) {};
\node[pointtt] (PPPY2) at ({(1)*\InterFeuilles},{\NiveauC}) {};
\node[pointtt] (PPPY3) at ({(1.2)*\InterFeuilles},{\NiveauC}) {};

\node[point] (P4) at ({(1.3)*\InterFeuilles},{\NiveauC-0.6}) {};
\node[number] at ({(1)*\InterFeuilles},{-12-0.5}) {$n_{1,0}$};
\node[point] (P44) at ({(0.7)*\InterFeuilles},{\NiveauC-0.6}) {};
\node[noeudvert] (C22) at ({(1.5)*\InterFeuilles},{\NiveauC}) {};

\draw[fleche] (A)--(B1) node[etiquette] {};
\draw[fleche] (A)--(B2) node[etiquette] {};
\draw[fleche] (A)--(B3) node[etiquette] {};
\draw[fleche] (A)--(B4) node[etiquette] {};
\draw[fleche] (B1)--(C11) node[etiquette] {};
\draw[fleche] (B1)--(C12) node[etiquette] {};
\draw[fleche] (B2)--(C21) node[etiquette] {};
\draw[fleche] (B2)--(C22) node[etiquette] {};

\draw[fleche2] (P1)--(P11) {};
\draw[fleche2] (P2)--(P22) {};
\draw[fleche2] (P3)--(P33) {};
\draw[fleche2] (P4)--(P44) {};

\def\shift{8}

\tikzstyle{noeud}=[fill=black!45!white,draw,circle,scale=0.5]
\tikzstyle{noeudrouge}=[fill=red!45!white,draw,circle,scale=0.5]
\tikzstyle{noeudbleu}=[fill=blue!45!white,draw,circle,scale=0.5]
\tikzstyle{noeudvert}=[fill=green!75!black!45!white,draw,circle,scale=0.5]

\node[noeudvert] (DAG3) at ({\shift+(5)*\InterFeuilles},{\NiveauC}) {0};
\node[noeudrouge] (DAG2) at ({\shift+(5)*\InterFeuilles},{\NiveauB}) {1};
\node[noeud] (DAG1) at ({\shift+(5)*\InterFeuilles},{\NiveauA}) {2};

\draw[fleche] (DAG2)--(DAG3) node[etiquette] {\footnotesize$n_{1,0}$\qquad$~$};
\draw[fleche] (DAG1)--(DAG2) node[etiquette] {\footnotesize$n_{2,1}$\qquad$~$};
\path[fleche] (DAG1) edge[bend left=60] node {\footnotesize$~$\qquad$n_{2,0}$} (DAG3);
\end{tikzpicture}}

\newcommand{\treeTC}{
\begin{tikzpicture}[xscale=0.15,yscale=0.15]
\tikzstyle{fleche}=[->,>=latex,thick]
\tikzstyle{noeud}=[fill=white,circle,draw]

\tikzstyle{noeudn}=[fill=black,circle,draw,scale=0.5]
\tikzstyle{noeudb}=[fill=blue,circle,draw,scale=0.5]
\tikzstyle{noeudr}=[fill=red,circle,draw,scale=0.5]
\tikzstyle{noeudg}=[fill=green!70!black,circle,draw,scale=0.5]
\tikzstyle{noeudo}=[fill=orange,circle,draw,scale=0.5]

\tikzstyle{etiquette}=[midway]

\def\DistanceInterNiveaux{5}
\def\DistanceInterFeuilles{5}

\def\NiveauA{(-0)*\DistanceInterNiveaux}
\def\NiveauB{(-1.1)*\DistanceInterNiveaux}
\def\NiveauC{(-2.2)*\DistanceInterNiveaux}
\def\NiveauD{(-3.3)*\DistanceInterNiveaux}
\def\InterFeuilles{(1)*\DistanceInterFeuilles}

\node[noeudn] (A) at ({(0)*\InterFeuilles},{\NiveauA}) {};

\node[noeudr] (B1) at ({(-1.2)*\InterFeuilles},{\NiveauB}) {};
\node[noeudo] (B2) at ({(1)*\InterFeuilles},{\NiveauB}) {};

\node[noeudg] (C11) at ({(-1.5)*\InterFeuilles},{\NiveauC}) {};
\node[noeudg] (C12) at ({(-0.8)*\InterFeuilles},{\NiveauC}) {};

\node[noeudr] (C21) at ({(0.2)*\InterFeuilles},{\NiveauC}) {};
\node[noeudb] (C22) at ({(1)*\InterFeuilles},{\NiveauC}) {};
\node[noeudg] (C23) at ({(1.8)*\InterFeuilles},{\NiveauC}) {};

\node[noeudg] (D211) at ({(0)*\InterFeuilles},{\NiveauD}) {};
\node[noeudg] (D212) at ({(0.4)*\InterFeuilles},{\NiveauD}) {};
\node[noeudg] (D221) at ({(1)*\InterFeuilles},{\NiveauD}) {};

\draw[fleche] (A)--(B1) node[etiquette] {};
\draw[fleche] (A)--(B2) node[etiquette] {};
\draw[fleche] (B1)--(C11) node[etiquette] {};
\draw[fleche] (B1)--(C12) node[etiquette] {};
\draw[fleche] (B2)--(C21) node[etiquette] {};
\draw[fleche] (B2)--(C22) node[etiquette] {};
\draw[fleche] (B2)--(C23) node[etiquette] {};
\draw[fleche] (C21)--(D211) node[etiquette] {};
\draw[fleche] (C21)--(D212) node[etiquette] {};
\draw[fleche] (C22)--(D221) node[etiquette] {};

\tikzstyle{fleche}=[->,>=latex,thick]

\tikzstyle{noeud}=[fill=black,draw,circle,scale=0.5]
\tikzstyle{noeudrouge}=[fill=red,draw,circle,scale=0.5]
\tikzstyle{noeudbleu}=[fill=blue,draw,circle,scale=0.5]
\tikzstyle{noeudvert}=[fill=green!70!black,draw,circle,scale=0.5]
\tikzstyle{noeudor}=[fill=orange,draw,circle,scale=0.5]

\tikzstyle{etiquette}=[midway]

\def\shift{8}

\node[noeudvert] (DAG5) at ({\shift+(4)*\InterFeuilles},{\NiveauD}) {};
\node[noeudbleu] (DAG4) at ({\shift+(4)*\InterFeuilles},{\NiveauC}) {};
\node[noeudrouge] (DAG3) at ({\shift+(3)*\InterFeuilles},{\NiveauC}) {};
\node[noeudor] (DAG2) at ({\shift+(4)*\InterFeuilles},{\NiveauB}) {};
\node[noeud] (DAG1) at ({\shift+(4)*\InterFeuilles},{\NiveauA}) {};

\path[fleche] (DAG2) edge[bend left=60] node {\footnotesize$\quad1$} (DAG5);
\draw[fleche] (DAG3)--(DAG5) node[etiquette] {\footnotesize$2\quad$};
\draw[fleche] (DAG4)--(DAG5) node[etiquette] {\footnotesize$\quad1$};
\draw[fleche] (DAG2)--(DAG4) node[etiquette] {\footnotesize$\quad1$};
\draw[fleche] (DAG2)--(DAG3) node[etiquette] {\footnotesize$1~\quad$};
\path[fleche] (DAG1) edge[bend left=-60] node {\footnotesize$\quad1$} (DAG3);
\draw[fleche] (DAG1)--(DAG2) node[etiquette] {\footnotesize$1\quad$};

\end{tikzpicture}}

\newcommand{\treeTCsn}{
\begin{tikzpicture}[xscale=0.15,yscale=0.15]
\tikzstyle{fleche}=[->,>=latex,thick]
\tikzstyle{noeud}=[fill=white,circle,draw]

\tikzstyle{noeudn}=[fill=black,circle,draw,scale=0.5]
\tikzstyle{noeudb}=[fill=blue,circle,draw,scale=0.5]
\tikzstyle{noeudr}=[fill=red,circle,draw,scale=0.5]
\tikzstyle{noeudg}=[fill=green!70!black,circle,draw,scale=0.5]
\tikzstyle{noeudo}=[fill=orange,circle,draw,scale=0.5]

\tikzstyle{etiquette}=[midway]

\def\DistanceInterNiveaux{5}
\def\DistanceInterFeuilles{5}

\def\NiveauA{(-0)*\DistanceInterNiveaux}
\def\NiveauB{(-1.1)*\DistanceInterNiveaux}
\def\NiveauC{(-2.2)*\DistanceInterNiveaux}
\def\NiveauD{(-3.3)*\DistanceInterNiveaux}
\def\InterFeuilles{(1)*\DistanceInterFeuilles}

\node[noeudn] (A) at ({(0)*\InterFeuilles},{\NiveauA}) {};

\node[noeudr] (B1) at ({(-1.2)*\InterFeuilles},{\NiveauB}) {};
\node[noeudo] (B2) at ({(1)*\InterFeuilles},{\NiveauB}) {};

\node[noeudg] (C11) at ({(-1.5)*\InterFeuilles},{\NiveauC}) {};
\node[noeudg] (C12) at ({(-0.8)*\InterFeuilles},{\NiveauC}) {};

\node[noeudr] (C21) at ({(0.2)*\InterFeuilles},{\NiveauC}) {};
\node[noeudr] (C22) at ({(1)*\InterFeuilles},{\NiveauC}) {};
\node[noeudg] (C23) at ({(1.8)*\InterFeuilles},{\NiveauC}) {};

\node[noeudg] (D211) at ({(0)*\InterFeuilles},{\NiveauD}) {};
\node[noeudg] (D212) at ({(0.4)*\InterFeuilles},{\NiveauD}) {};
\node[noeudg] (D221) at ({(0.8)*\InterFeuilles},{\NiveauD}) {};
\node[noeudg] (D222) at ({(1.2)*\InterFeuilles},{\NiveauD}) {};

\draw[fleche] (A)--(B1) node[etiquette] {};
\draw[fleche] (A)--(B2) node[etiquette] {};
\draw[fleche] (B1)--(C11) node[etiquette] {};
\draw[fleche] (B1)--(C12) node[etiquette] {};
\draw[fleche] (B2)--(C21) node[etiquette] {};
\draw[fleche] (B2)--(C22) node[etiquette] {};
\draw[fleche] (B2)--(C23) node[etiquette] {};
\draw[fleche] (C21)--(D211) node[etiquette] {};
\draw[fleche] (C21)--(D212) node[etiquette] {};
\draw[fleche] (C22)--(D222) node[etiquette] {};
\draw[fleche] (C22)--(D221) node[etiquette] {};

\tikzstyle{fleche}=[->,>=latex,thick]

\tikzstyle{noeud}=[fill=black,draw,circle,scale=0.5]
\tikzstyle{noeudrouge}=[fill=red,draw,circle,scale=0.5]
\tikzstyle{noeudbleu}=[fill=blue,draw,circle,scale=0.5]
\tikzstyle{noeudvert}=[fill=green!70!black,draw,circle,scale=0.5]
\tikzstyle{noeudor}=[fill=orange,draw,circle,scale=0.5]

\tikzstyle{etiquette}=[midway]

\def\shift{8}

\node[noeudvert] (DAG5) at ({\shift+(4)*\InterFeuilles},{\NiveauD}) {};
\node[noeudrouge] (DAG3) at ({\shift+(4)*\InterFeuilles},{\NiveauC}) {};
\node[noeudor] (DAG2) at ({\shift+(4)*\InterFeuilles},{\NiveauB}) {};
\node[noeud] (DAG1) at ({\shift+(4)*\InterFeuilles},{\NiveauA}) {};

\path[fleche] (DAG2) edge[bend left=60] node {\footnotesize$\quad1$} (DAG5);
\draw[fleche] (DAG3)--(DAG5) node[etiquette] {\footnotesize$2\quad$};
\draw[fleche] (DAG2)--(DAG3) node[etiquette] {\footnotesize$2\quad$};
\path[fleche] (DAG1) edge[bend left=-60] node {\footnotesize$1\quad$} (DAG3);
\draw[fleche] (DAG1)--(DAG2) node[etiquette] {\footnotesize$1\quad$};

\end{tikzpicture}}

%% file: tikz_ex-worstcase.tex
\newcommand{\wcdeuxtrois}{
\begin{tikzpicture}[xscale=0.3,yscale=0.3]
\tikzstyle{fleche}=[-,>=latex]
\tikzstyle{noeud}=[draw,circle,fill=black,scale=0.2]
\node[noeud] (N) at ({0},{0}) {};

\node[noeud] (N00) at ({-2.0},{-2.0}) {};
\node[noeud] (N01) at ({0.0},{-2.0}) {};
\node[noeud] (N02) at ({2.0},{-2.0}) {};

\draw[fleche] (N)--(N00) {};
\draw[fleche] (N)--(N01) {};
\draw[fleche] (N)--(N02) {};

\node[noeud] (N000) at ({-2.0},{-3.7}) {};
\node[noeud] (N010) at ({0.0},{-3.7}) {};
\draw[fleche] (N00)--(N000) {};
\draw[fleche] (N01)--(N010) {};

\node[noeud] (N021) at ({1.0},{-3.7}) {};
\node[noeud] (N020) at ({2.0},{-3.7}) {};
\node[noeud] (N022) at ({3.0},{-3.7}) {};

\draw[fleche] (N02)--(N020) {};
\draw[fleche] (N02)--(N021) {};
\draw[fleche] (N02)--(N022) {};

\end{tikzpicture}
}

\newcommand{\wcdeuxtroisSN}{
\begin{tikzpicture}[xscale=0.3,yscale=0.3]
\tikzstyle{fleche}=[-,>=latex]
\tikzstyle{noeud}=[draw,circle,fill=black,scale=0.2]
\node[noeud] (N) at ({0},{0}) {};

\node[noeud] (N00) at ({-2.0},{-2.0}) {};
\node[noeud] (N01) at ({0.0},{-2.0}) {};
\node[noeud] (N02) at ({2.0},{-2.0}) {};

\draw[fleche] (N)--(N00) {};
\draw[fleche] (N)--(N01) {};
\draw[fleche] (N)--(N02) {};

\node[noeud] (N000) at ({-2.0},{-3.7}) {};
\node[noeud] (N010) at ({0.0},{-3.7}) {};
\draw[fleche] (N00)--(N000) {};
\draw[fleche] (N01)--(N010) {};

\node[noeud] (N020) at ({2.0},{-3.7}) {};
\draw[fleche] (N02)--(N020) {};
\end{tikzpicture}
}

\newcommand{\wcdeuxquatre}{
\begin{tikzpicture}[xscale=0.3,yscale=0.3]
\tikzstyle{fleche}=[-,>=latex]
\tikzstyle{noeud}=[draw,circle,fill=black,scale=0.2]
\node[noeud] (N) at ({0},{0}) {};

\node[noeud] (N00) at ({-3.0},{-2.0}) {};
\node[noeud] (N01) at ({-1.0},{-2.0}) {};
\node[noeud] (N02) at ({1.0},{-2.0}) {};
\node[noeud] (N03) at ({3.0},{-2.0}) {};

\draw[fleche] (N)--(N00) {};
\draw[fleche] (N)--(N01) {};
\draw[fleche] (N)--(N02) {};
\draw[fleche] (N)--(N03) {};

\node[noeud] (N000) at ({-3.5},{-3.7}) {};
\node[noeud] (N001) at ({-2.5},{-3.7}) {};
\draw[fleche] (N00)--(N000) {};
\draw[fleche] (N00)--(N001) {};

\node[noeud] (N010) at ({-1.5},{-3.7}) {};
\node[noeud] (N011) at ({-0.5},{-3.7}) {};
\draw[fleche] (N01)--(N010) {};
\draw[fleche] (N01)--(N011) {};

\node[noeud] (N020) at ({0.4},{-3.7}) {};
\node[noeud] (N021) at ({1.6},{-3.7}) {};
\node[noeud] (N022) at ({0.8},{-3.7}) {};
\node[noeud] (N023) at ({1.2},{-3.7}) {};
\draw[fleche] (N02)--(N020) {};
\draw[fleche] (N02)--(N021) {};
\draw[fleche] (N02)--(N022) {};
\draw[fleche] (N02)--(N023) {};

\node[noeud] (N030) at ({2.4},{-3.7}) {};
\node[noeud] (N031) at ({3.6},{-3.7}) {};
\node[noeud] (N032) at ({2.8},{-3.7}) {};
\node[noeud] (N033) at ({3.2},{-3.7}) {};
\draw[fleche] (N03)--(N030) {};
\draw[fleche] (N03)--(N031) {};
\draw[fleche] (N03)--(N032) {};
\draw[fleche] (N03)--(N033) {};

\end{tikzpicture}
}

\newcommand{\wcdeuxquatreSN}{
\begin{tikzpicture}[xscale=0.3,yscale=0.3]
\tikzstyle{fleche}=[-,>=latex]
\tikzstyle{noeud}=[draw,circle,fill=black,scale=0.2]
\node[noeud] (N) at ({0},{0}) {};

\node[noeud] (N00) at ({-3.0},{-2.0}) {};
\node[noeud] (N01) at ({-1.0},{-2.0}) {};
\node[noeud] (N02) at ({1.0},{-2.0}) {};
\node[noeud] (N03) at ({3.0},{-2.0}) {};

\draw[fleche] (N)--(N00) {};
\draw[fleche] (N)--(N01) {};
\draw[fleche] (N)--(N02) {};
\draw[fleche] (N)--(N03) {};

\node[noeud] (N000) at ({-3.5},{-3.7}) {};
\node[noeud] (N001) at ({-2.5},{-3.7}) {};
\draw[fleche] (N00)--(N000) {};
\draw[fleche] (N00)--(N001) {};

\node[noeud] (N010) at ({-1.5},{-3.7}) {};
\node[noeud] (N011) at ({-0.5},{-3.7}) {};
\draw[fleche] (N01)--(N010) {};
\draw[fleche] (N01)--(N011) {};

\node[noeud] (N020) at ({0.5},{-3.7}) {};
\node[noeud] (N021) at ({1.5},{-3.7}) {};
\draw[fleche] (N02)--(N020) {};
\draw[fleche] (N02)--(N021) {};

\node[noeud] (N030) at ({2.5},{-3.7}) {};
\node[noeud] (N031) at ({3.5},{-3.7}) {};
\draw[fleche] (N03)--(N030) {};
\draw[fleche] (N03)--(N031) {};

\end{tikzpicture}
}